\documentclass[pdflatex,sn-vancouver-ay]{sn-jnl}

\usepackage{graphicx}%
\usepackage{subcaption}
\usepackage{multirow}%
\usepackage{amsmath,amssymb,amsfonts}%
\usepackage{amsthm}%
\usepackage{mathrsfs}%
\usepackage[title]{appendix}%
\usepackage{xcolor}%
\usepackage{textcomp}%
\usepackage{manyfoot}%
\usepackage{booktabs}%
\usepackage{algorithm}%
\usepackage{algorithmicx}%
\usepackage{algpseudocode}%
\usepackage{listings}%

\theoremstyle{thmstyleone}%
\newtheorem{theorem}{Theorem}
\newtheorem{lemma}[theorem]{Lemma}%

\theoremstyle{thmstyletwo}%

\theoremstyle{thmstylethree}%

\begin{document}

\title[Sequential Experimental Designs for Kriging Model]{Sequential Experimental Designs for Kriging Model}


\author[1]{\fnm{Ruonan} \sur{Zheng}}\email{rnanzheng@163.com}

\author[1]{\fnm{Min-Qian} \sur{Liu}}\email{mqliu@nankai.edu.cn}

\author*[1]{\fnm{Yongdao} \sur{Zhou}}\email{ydzhou@nankai.edu.cn}

\author*[2]{\fnm{Xuan} \sur{Chen}}\email{chenxuan@nudt.edu.cn}

\affil[1]{\orgdiv{NITFID, LPMC$\&$KLMDASR, School of Statistics and Data Science}, \orgname{Nankai University}, \orgaddress{\city{Tianjin}, \postcode{300071}, \country{China}}}

\affil[2]{\orgdiv{College of Sciences}, \orgname{National University of Defense Technology}, \orgaddress{\city{Changsha}, \postcode{410073}, \country{China}}}


\abstract{Computer experiments have become an indispensable alternative to complex physical and engineering experiments. The Kriging model is the most widely used surrogate model, with the core goal of minimizing the discrepancy between the surrogate and true models across the entire experimental domain. However, existing sequential design methods have critical limitations: observation-based batch sequential designs are rarely studied, while one-point sequential designs have insufficient information utilization and suffer from inefficient resource utilization---they require numerous repeated observation rounds to accumulate sufficient points, leading to prolonged experimental cycles. To address these gaps, this paper proposes two novel one-point sequential design criteria and a general batch sequential design framework. Moreover, the batch sequential design framework solves the inherent point clustering problem in naive batch selection, enabling efficient extension of any sequential criterion to batch scenarios. Simulations on some test functions demonstrate that the proposed methods outperform existing approaches in terms of fitting accuracy in most cases.}

\keywords{batch design, Computer experiments, Global fitting, Surrogate model.}



\maketitle

\section{Introduction}
\label{sec:1}
Computer experiments find extensive applications in the real world, especially as alternatives to certain complex physical and engineering experiments \citep{Santner:2003}. The Gaussian process model \citep{Gramacy:2020}, also known as the Kriging model, is the most prevalently used in computer experiments. When constructing a Kriging model, the selection of experimental points is crucial. We term all experimental points as a design. Therefore, fitting models with different designs will yield models of varying performances. In general, obtaining observations from the complex model in the real world is extremely resource-consuming. Therefore, the objective is to leverage existing information as much as possible to select fewer observation points, ensuring that the model fitted based on the observed data exhibits optimal performance. When selecting an experimental design, a common practice is to determine all experimental points at once and observe them. In this case, we typically choose a space-filling design which can uniformly explore the design space for observation \citep{Fang:2006}. There are numerous criteria for measuring the space-filling properties of a design, and we can choose to implement a design that performs well under one of these criteria. For example, discrepancies \citep{Fang:2018} characterize the space-filling properties by calculating the difference between the empirical distribution of a design and the uniform distribution. 

In contrast to single-stage experiments where all experimental points are pre-determined in a single step, relying entirely on a priori assumptions about the experimental domain and underlying true model---sequential designs offer a dynamic and adaptive alternative by enabling the iterative addition of new design points to an existing set based on accumulated observational information, addressing the critical limitations of single-stage designs. If pre-selected points in single-stage designs fail to cover key regions, the resulting model will suffer from poor fitting performance, which leads to wasted resources and suboptimal outcomes, whereas sequential designs leverage data-driven feedback to refine the design set iteratively---each new round of observations reveals the model's fitting gaps, guiding subsequent points to target under-explored or high-uncertainty regions, thus not only enhancing global fitting accuracy but also improving resource efficiency by avoiding redundant observations in low-value regions and concentrating efforts on high-impact areas. The construction of sequential designs is tailored to specific research goals. For instance, \cite{Schonlau:1998} proposed the expected improvement (EI) algorithm for the sequential search for the global minimum. \cite{Lam:2008} developed a modified EI criterion for global fit. \cite{Ranjan:2008} put forward a sequential methodology for estimating a contour from a complex computer code. 

Sequential experiments can be broadly classified into two distinct types. The first category is observation-free sequential design, in which new design points added at each stage are pre-determined without incorporating data from prior observations. For example, \cite{Kim:2017} introduced a sequential design based on the minimum energy criterion. \cite{Zhang:2024} constructed sequential designs with good space-filling properties through the good lattice point method. The second category is observation-based sequential design, a data-driven approach that necessitates observing the response of each design point and utilizing the acquired information to guide the selection of subsequent experimental points. For example, \cite{Mu:2016} refined the criterion proposed by \cite{Lam:2008} to develop sequential designs.

However, while existing research has yielded several sequential design criteria, most are tailored to one-point sequential designs. Such one-point criteria tend to incur excessive iterative cycles, leading to inefficient resource utilization. In many practical scenarios, batch sequential designs are far more necessary. For instance, in high-fidelity aerospace component simulations, each round of observation requires substantial fixed costs including equipment calibration, simulation environment setup, model initialization, and parameter tuning, whereas the marginal cost of increasing the number of observation points within the same round is negligible. In such cases, batch sequential designs offer distinct advantages: they significantly reduce the number of experimental rounds, avoiding redundant fixed cost consumption, and enhance the efficiency of model fitting by supplementing multiple key regions simultaneously. Despite this, research on observation-based batch sequential designs remains scarce. Aiming at the goal of global fitting, this paper first puts forward two novel criteria for selecting sequential design points: a gradient-based criterion and a variance-based criterion. Additionally, by combining observation-based sequential design criteria with batch sequential designs, we introduce a framework of batch sequential design that can be applied to any sequential design criterion. Finally, through simulations, we demonstrate that the sequential designs derived from the proposed criteria and framework generally outperform.

The structure of this paper is organized as follows. Section \ref{sec:2} outlines the notation and background related to Kriging models and sequential designs. Section \ref{sec:3} presents two novel sequential design criteria. Section \ref{sec:4} introduces a framework for generating batch sequential designs. Section \ref{sec:5} displays the simulation results of comparisons between the proposed methods and existing approaches. Section \ref{sec:conc} concludes this paper. 

\section{Notation and background}
\label{sec:2}
In this section, we first introduce the notations relevant to the true model and review the corresponding modeling approaches. Furthermore, we outline the framework for sequential design and summarize existing sequential design criteria.

\subsection{Kriging model}
Let $Z(\boldsymbol{x})$ be a Gaussian process on $R^m$. Assume its covariance structure is given by
$$Cov(Z(\boldsymbol{x}),Z(\boldsymbol{y}))=\tau^2k(\boldsymbol{x},\boldsymbol{y})$$ for any $\boldsymbol{x},\boldsymbol{y}\in R^m$, where $\tau^2$ presents the variance and $k(\boldsymbol{x},\boldsymbol{y})$ is the correlation function. Several special correlation functions are frequently employed. For example, the separable Gaussian correlation function is defined as
$$k(\boldsymbol{x},\boldsymbol{y};\theta)=\exp\left\{-\sum_{k=1}^{m}\theta_k|x_{k}-y_{k}|^2\right\}$$ and the Mat\'{e}rn correlation function is 
$$k(\boldsymbol{x},\boldsymbol{y};\nu,\phi)=\frac{1}{\Gamma(\nu) 2^{\nu-1}} (2\sqrt{\nu} \phi \|\boldsymbol{x}-\boldsymbol{y}\|)^\nu K_\nu (2\sqrt{\nu} \phi \|\boldsymbol{x}-\boldsymbol{y}\|),$$
where $\boldsymbol{x}=(x_1,\ldots,x_m)$, $\boldsymbol{y}=(y_1,\ldots,y_m)$, $\|\boldsymbol{x}\|$ denotes the $l_2$-norm of vector $\boldsymbol{x}$, and $K_\nu$ is the modified Bessel function of the second kind.

When modeling the relationship between input variables and responses, Gaussian process modeling is a prevalent approach, and the resulting model is also known as the Kriging model. Let $\boldsymbol{X}=(\boldsymbol{x}_1,\ldots,\boldsymbol{x}_n)^T$ and $\boldsymbol{Y}_{\boldsymbol{X}}=(Z(\boldsymbol{x}_1),\ldots,Z(\boldsymbol{x}_n))$ be the observations at $\boldsymbol{X}$. Based on $\boldsymbol{X}$ and $\boldsymbol{Y}_{\boldsymbol{X}}$, we can compute the prediction and its variance at an unobserved point $\boldsymbol{x}$. The specific expressions are
\begin{align*}
	\hat{y}_{\boldsymbol{X}}(\boldsymbol{x})&=\boldsymbol{r}_{\boldsymbol{X}}^T(\boldsymbol{x})\boldsymbol{K}^{-1}(\boldsymbol{X})\boldsymbol{Y}_{\boldsymbol{X}},\\
	\widehat{s^2_{\boldsymbol{X}}}(\boldsymbol{x})&=\tau^2\left(k(\boldsymbol{x},\boldsymbol{x})-\boldsymbol{r}_{\boldsymbol{X}}^T(\boldsymbol{x})\boldsymbol{K}^{-1}(\boldsymbol{X})\boldsymbol{r}_{\boldsymbol{X}}(\boldsymbol{x})\right),
\end{align*}
where $\boldsymbol{r}_{\boldsymbol{X}}(\boldsymbol{x})=(k(\boldsymbol{x},\boldsymbol{x}_1),\ldots,k(\boldsymbol{x},\boldsymbol{x}_n))^T$ and $\boldsymbol{K}(\boldsymbol{X})=(k(\boldsymbol{x}_i,\boldsymbol{x}_j))_{n\times n}$. The subscript \( \boldsymbol{X} \) signifies that the value is derived or observed based on \( \boldsymbol{X} \). It can be proven that $\hat{y}_{\boldsymbol{X}}(\boldsymbol{x})$ is the best linear predictor based on $\boldsymbol{Y}_{\boldsymbol{X}}$, and both $\hat{y}_{\boldsymbol{X}}(\boldsymbol{x})$ and $\widehat{s_{\boldsymbol{X}}^2}(\boldsymbol{x})$ have explicit expressions. These aspects constitute important reasons why the Kriging model ranks among the most frequently used models.

Our work in this paper focuses on global fitting of the Kriging model. However, the model's performance is highly dependent on the selection of observation points $\boldsymbol{x}_1,\ldots,\boldsymbol{x}_n$. Consequently, the strategic selection of these points proves to be a critical determinant of fitting accuracy. In the subsequent subsection, we will outline a framework for systematically selecting design points to address this key requirement. 

\subsection{Sequential design framework}
Space-filling designs are widely adopted for fitting computer experiments, where each design point corresponds to an observation point in the experimental domain. In general, a single-stage space-filling design is employed to construct the Kriging model, meaning that all design points to be observed are pre-determined prior to the execution of any experimental observations. A variety of criteria exist for evaluating the space-filling performance of a design, such as discrepancy measures and the maximin distance criterion. For practical implementation, we can directly select a design that exhibits superior performance under a specific space-filling criterion as the single-stage design for Kriging model fitting.

However, in certain practical scenarios, such as when computational resources are constrained, it is necessary to select design points for observation in a sequential manner, which is termed a sequential design. Let \(\boldsymbol{X}_0\) be the initial design, \(\boldsymbol{B}_i=(\boldsymbol{x}_{i1},\ldots,\boldsymbol{x}_{ib})^T\) be the design added at the \(i\)-th iteration where $b$ is the batch size, and \(\boldsymbol{X}_i\) be the design composed of all points after the \(i\)-th iteration, i.e., $\boldsymbol{X}_i = (\boldsymbol{X}^T_0, \boldsymbol{B}^T_1, \ldots, \boldsymbol{B}^T_i)^T$. Let \(\boldsymbol{Y}_i\) be the observed values of the points in \(\boldsymbol{X}_i\). In practical applications, only the points in $\boldsymbol{B}_i$ need to be observed at each iteration, rather than re-observing the entire set of points in $\boldsymbol{X}_i$. The pseudocode for the framework of constructing Kriging model-based sequential designs is presented in Algorithm \ref{alg:1}.

\begin{algorithm}[htbp]
	\caption{Sequential design for Kriging model}
	\label{alg:1}
	\begin{algorithmic}[1]
		\State Select a space-filling design as the initial design \(\boldsymbol{X}_0\).
		\State \(i \leftarrow 0\).
		\While{termination condition is not satisfied}
		\State \(\boldsymbol{Y}_{\boldsymbol{X}_i}\) $\leftarrow$ the observed values of \(\boldsymbol{X}_i\).
		\State Select \(\boldsymbol{B}_{i+1}\) according to the criterion \(\phi\).
		\State \(\boldsymbol{X}_{i+1} \leftarrow (\boldsymbol{X}_i, \boldsymbol{B}_{i+1})\).
		\State \(i \leftarrow i + 1\).
		\EndWhile
		\State Fit the Kriging model using the final \(\boldsymbol{X}_i\) and \(\boldsymbol{Y}_{\boldsymbol{X}_i}\).
	\end{algorithmic}
\end{algorithm}

In Algorithm 1, the number of rows in the initial design, the number of rows in $\boldsymbol{B}_i$, the criterion $\phi$, and the termination condition must all be determined according to practical scenarios. Specifically, $\boldsymbol{B}_i$ may consist of a single design point or multiple design points. If the criterion $\phi$ is independent of $\boldsymbol{Y}_{\boldsymbol{X}_i}$, it is unnecessary to conduct observations for $\boldsymbol{B}_i$ in each iteration; otherwise, it is essential to refit the Kriging model using $\boldsymbol{X}_i$ and $\boldsymbol{Y}_{\boldsymbol{X}_i}$, as the model parameters will change correspondingly. In addition, the termination condition can be defined in various ways, such as whether the fitting accuracy of the Kriging model has reached a preset threshold, or whether the total number of observation points has hit a predefined upper limit.

\subsection{Review of sequential design criteria}
In this part, we review some criteria for sequential designs. Generally speaking, for a given criterion $\phi$, the next point $\widetilde{\boldsymbol{x}}$ to be added is the point that maximizes $\phi(\boldsymbol{x})$ in the experimental region $\mathcal{X}$, which is formally expressed as
$$\widetilde{\boldsymbol{x}}=\arg\max_{\boldsymbol{x}\in\mathcal{X}}\phi(\boldsymbol{x}).$$
We classify these criteria into two distinct categories: the first category includes criteria where all sequential design points are pre-determined before observing the response values of any experimental points, which we term observation-free criteria; the second category encompasses criteria that rely on existing design points and their corresponding response values for sequential point selection, which we refer to as observation-based criteria.

For convenience, we denote the existing design as $\boldsymbol{X}$. For the observation-free criteria, some commonly used space-filling criteria, such as the maximin distance criterion and discrepancies, can be adopted. The core objective is to ensure that the augmented design comprised of the newly selected design point and the existing design points exhibits favorable space-filling properties within the experimental domain. In this paper, we employ the mixture discrepancy (MD) as the space-filling criterion. Since the performance of other space-filling criteria is similar to that of MD, they will not be considered herein. The MD of a design $\boldsymbol{D}=(x_{ij})_{n\times m}$ is
\begin{align*}
	MD^2(\boldsymbol{D})=&\left( \frac{19}{12} \right)^m - \frac{2}{n} \sum_{i=1}^n \prod_{j=1}^m \left( \frac{5}{3} - \frac{1}{4} \left| \tilde{x}_{ij} \right| - \frac{1}{4} \left| \tilde{x}_{ij} \right|^2 \right)\\ &+ \frac{1}{n^2} \sum_{i,k=1}^n\prod_{j=1}^m \left( \frac{15}{8} - \frac{1}{4} \left| \tilde{x}_{ij} \right| - \frac{1}{4} \left| \tilde{x}_{kj} \right| - \frac{3}{4} \left| \tilde{x}_{ij} - \tilde{x}_{kj} \right| + \frac{1}{2} \left| \tilde{x}_{ij} - \tilde{x}_{kj} \right|^2 \right),
\end{align*}
where $\tilde{x}_{ij}=x_{ij}-0.5$ for $i=1,\ldots,n$ and $j=1,\ldots,m$. For more details about MD, one can refer to \cite{Zhou:2013}. Under the MD criterion, the next design point $\tilde{x}$ to be added is selected by solving 
$$\widetilde{\boldsymbol{x}}=\arg\max_{\boldsymbol{x}\in\mathcal{X}}-MD^2(\boldsymbol{X}\cup\boldsymbol{x}).$$

For the observation-based criteria, \cite{Sacks:1988} proposed a strategy where one selects the point that minimizes the maximum $\widehat{s_{\boldsymbol{X}}^2}(\boldsymbol{x})$ for all $\boldsymbol{x}\in\mathcal{X}$. However, this strategy may consume much computational time. Therefore, we simplify it to finding the point with the maximum variance based on the existing design, that is, 
$$\widetilde{\boldsymbol{x}}=\arg\max_{\boldsymbol{x}\in\mathcal{X}}\phi_s(\boldsymbol{x}),$$
where $\phi_s(\boldsymbol{x})=\widehat{s_{\boldsymbol{X}}^2}(\boldsymbol{x})$. Assume that the true model between the input and the response is \( f(\boldsymbol{x}) \) and it is smooth. In practice, the specific form of $f(\boldsymbol{x})$ is typically unknown, and we can only obtain the true response values at a finite set of design points. Let $\boldsymbol{x}^*=\arg\min_{\boldsymbol{y}\in\boldsymbol{X}}d(\boldsymbol{x},\boldsymbol{y})$ denotes the point in the existing design $\boldsymbol{X}$ that is closest to $\boldsymbol{x}$ and $d(\boldsymbol{x},\boldsymbol{y})$ represents the Euclidean distance between $\boldsymbol{x}$ and $\boldsymbol{y}$. Building on this, \cite{Lam:2008} defined the improvement function as $I_0(\boldsymbol{x})=\left(f(\boldsymbol{x})-f(\boldsymbol{x}^*)\right)^2$. Let 
$$\phi_{EI_0}(\boldsymbol{x})=EI_0(\boldsymbol{x})=\left(\hat{y}_{\boldsymbol{X}}(\boldsymbol{x})-f(\boldsymbol{x}^*)\right)^2+\widehat{s_{\boldsymbol{X}}^2}(\boldsymbol{x}).$$
They proposed the expected improvement criterion for global fitting to choose the next point by maximizing $\phi_{EI_0}$. Specifically, the first term in $\phi_{EI_0}(\boldsymbol{x})$ quantifies the bias of the Kriging model prediction, while the second term represents the variance, which reflects the systematic uncertainty associated with the current observational information. \cite{Mu:2016} noted that $\phi_{EI_0}(\boldsymbol{x})$ fails to incorporate local directional information of the true function. To address this limitation, they revised the improvement function to $I_1(\boldsymbol{x})=\left(f(\boldsymbol{x})-f(\boldsymbol{x}^*)-\nabla^T f(\boldsymbol{x}^*)(\boldsymbol{x}-\boldsymbol{x}^*)\right)^2$ and proposed the new criterion
$$\phi_{EI_1}(\boldsymbol{x})=EI_1(\boldsymbol{x})=\left(\hat{y}_{\boldsymbol{X}}(\boldsymbol{x})-f(\boldsymbol{x}^*)-\nabla^T f(\boldsymbol{x}^*)(\boldsymbol{x}-\boldsymbol{x}^*)\right)^2+\widehat{s_{\boldsymbol{X}}^2}(\boldsymbol{x}).$$
Furthermore, to avoid selecting new points that are excessively close to existing design points, they introduced a distance penalty term and derived a variant of $\phi_{EI_1}(\boldsymbol{x})$ as
$$\phi_{EI_2}(\boldsymbol{x})=\left(\hat{y}_{\boldsymbol{X}}(\boldsymbol{x})-f(\boldsymbol{x}^*)-\nabla^T f(\boldsymbol{x}^*)(\boldsymbol{x}-\boldsymbol{x}^*)\right)^2\times d(\boldsymbol{x},\boldsymbol{x}^*)+\widehat{s_{\boldsymbol{X}}^2}(\boldsymbol{x}).$$
In subsequent sections of this paper, we will compare the method proposed in this study with the aforementioned criteria, elaborating on their respective advantages and disadvantages, applicable scenarios, and some relevant theoretical findings.

\section{Two novel criteria for sequential design}
\label{sec:3}
In this section, we propose two novel criteria for sequential design: the gradient-based criterion and the variance-based criterion. These two criteria attain the objective of global fitting from distinct perspectives, each addressing this core objective through a unique methodological approach.
\subsection{Gradient-based criterion}
We first elaborate on the motivation behind the development of the gradient-based criterion, followed by a detailed derivation of its mathematical formulation.
\subsubsection{Motivation}
We assume the true response model is
\begin{equation}
	\label{eq:2}
\begin{aligned}
	f(x_1,x_2)=&\left(15x_2 - \frac{5.1}{4\pi^2}(15x_1-5)^2 + \frac{5}{\pi}(15x_1-5) - 6\right)^2 \\&+ 10\left(1 - \frac{1}{8\pi}\right)\cos(15x_1-5) + 10,
\end{aligned}
\end{equation}
where $(x_1,x_2)\in\left[0,1\right]^2$. The function $f$ in Equation \eqref{eq:2} is the well-known Branin function \citep{Jamil:2013}, which is widely used as a test function in the field of surrogate modeling and experimental design.

To fit this true response model $f$, we adopt the Kriging model and select a separable Gaussian correlation function with a minor modification to avoid the singularity of the correlation matrix $\boldsymbol{K}$ \citep{Gramacy:2020}, that is
\begin{equation}
	\label{eq:6}
	k(\boldsymbol{x},\boldsymbol{y};\theta)=\exp\left\{-\sum_{k=1}^{2}\theta_k|x_{k}-y_{k}|^2\right\}+g\delta(\boldsymbol{x},\boldsymbol{y}),
\end{equation}
where $\theta_1,\theta_2,g$ are parameters and the indicator function $\delta(\boldsymbol{x},\boldsymbol{y})$ is defined as  $\delta(\boldsymbol{x},\boldsymbol{y})=1$ if $\boldsymbol{x}=\boldsymbol{y}$, and $\delta(\boldsymbol{x},\boldsymbol{y})=0$ otherwise.

Two distinct experimental designs are employed in this study. The first design, $\boldsymbol{D}_1$, is a uniform design optimized under the MD criterion. A uniform design with the specified number of rows and columns can be generated by the ``GenUD'' function from the ``UniDOE'' package \citep{Zhang:2018} in R, where we choose the number of levels to be the same as the number of rows. The second design $\boldsymbol{D}_2$ is the one-point sequential design with $10$ initial design points based on the gradient-based criterion introduced in the next subsection, which adds more points in regions with large gradient magnitudes while preserving favorable space-filling properties of the design.

\begin{figure}[htbp]
	\centering  
	\begin{subfigure}[b]{0.48\textwidth} 
		\centering  
		\includegraphics[width=\textwidth]{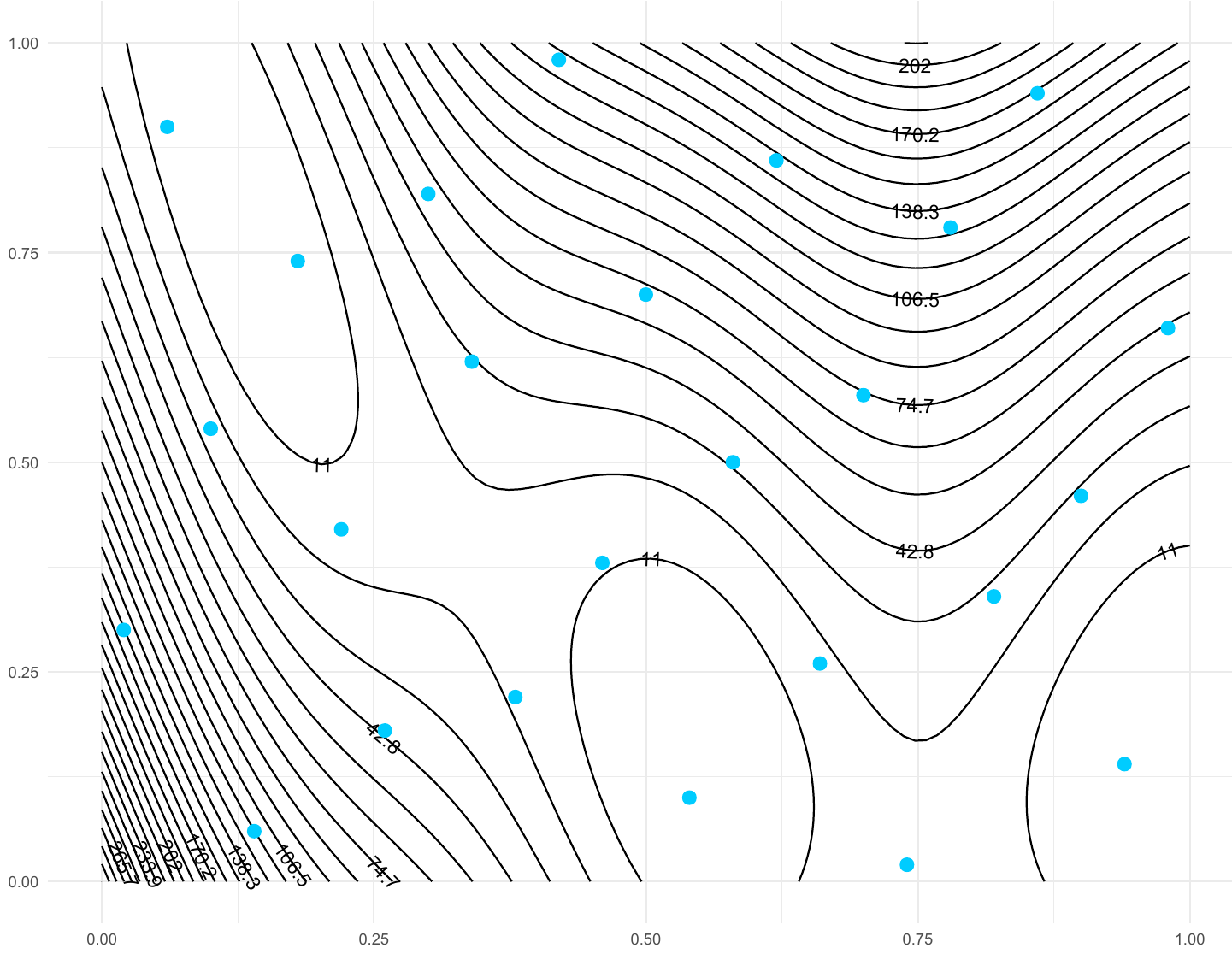}
		\caption{Design points of $\boldsymbol{D}_1$}
		\label{fig:sub1}
	\end{subfigure}
	\hfill  
	\begin{subfigure}[b]{0.48\textwidth} 
		\centering
		\includegraphics[width=\textwidth]{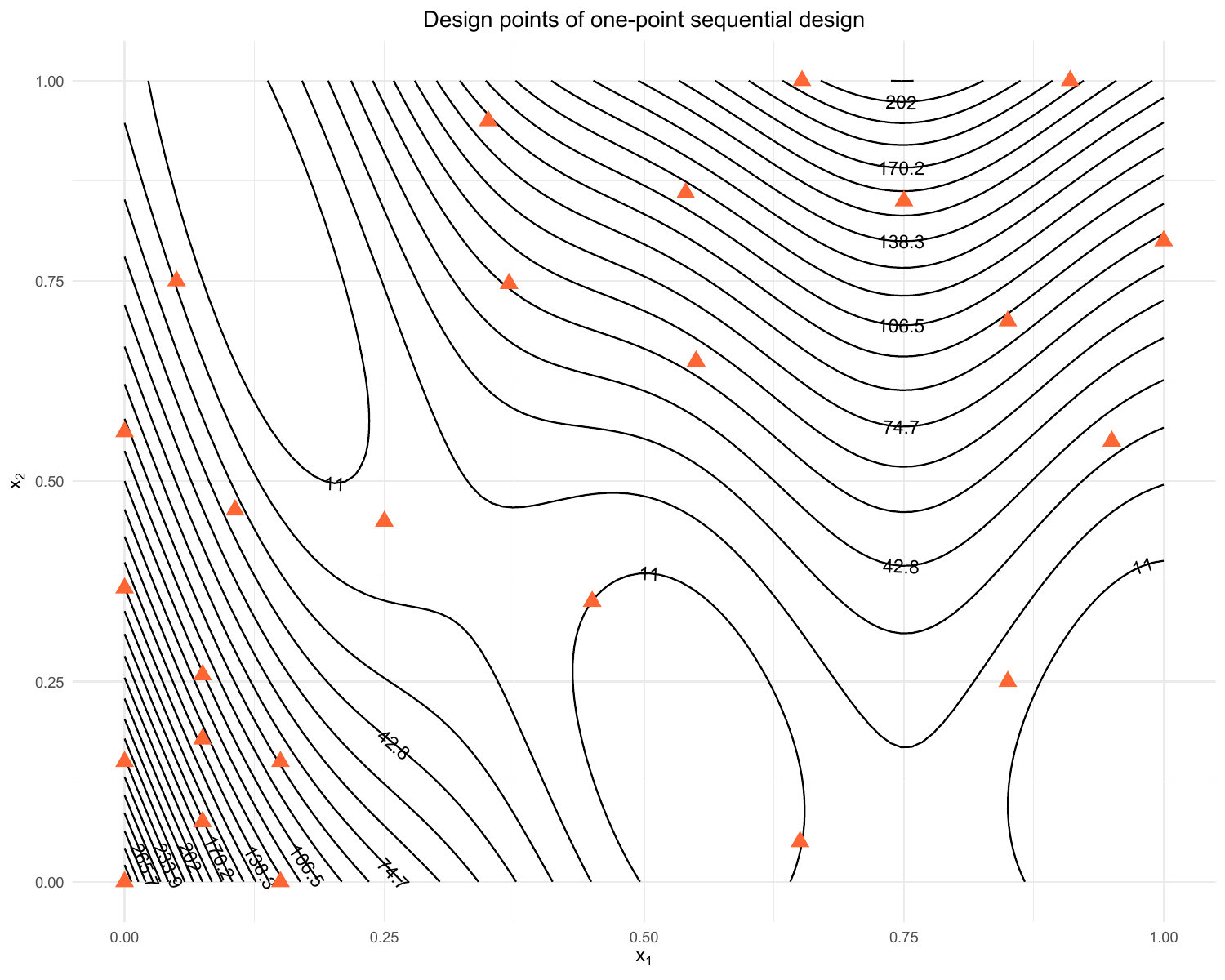}
		\caption{Design points of $\boldsymbol{D}_2$}
		\label{fig:sub2}
	\end{subfigure}
	\caption{The distribution of design points in the uniform design and gradient-based sequential design.} 
	\label{fig:1}
\end{figure}

Figure \ref{fig:1} shows the response surface of the function \( f(x_1,x_2) \) defined in Equation \eqref{eq:2}, along with the points distributions of the uniform design $\boldsymbol{D}_1$ and the one-point sequential design $\boldsymbol{D}_2$, both consisting of $25$ design points. Subplot (a) of Figure \ref{fig:1} depicts the design points of the uniform design, which are denoted by circular markers. Subplot (b) of Figure \ref{fig:2} illustrates the design points of the one-point sequential design, represented by triangular markers. It can be clearly observed that the design points of \(\boldsymbol{D}_2\) are more densely distributed in regions with large gradient magnitudes, while still maintaining favorable space-filling properties of the overall design.

Figure \ref{fig:2} presents the root mean square error (RMSE) of the uniform design and the one-point sequential design for the numbers of points ranging from $11$ to $25$. For each case, the uniform design is directly regenerated, while the one-point sequential design is sequentially generated based on its previous design points. To quantify the fitting accuracy of the two designs, we calculate the RMSE using a $10000\times2$ test matrix \(\boldsymbol{X}_{\text{test}}\) generated directly via the ``randomLHS'' function from the ``lhs'' package in R, that is, 
$$RMSE=\sqrt{\frac{\sum_{\boldsymbol{x}\in\boldsymbol{X}_{test}}\left(f(\boldsymbol{x})-\hat{y}(\boldsymbol{x})\right)^2}{10000}}.$$

\begin{figure}[htbp]
	\centering
	\includegraphics[width=0.8\textwidth]{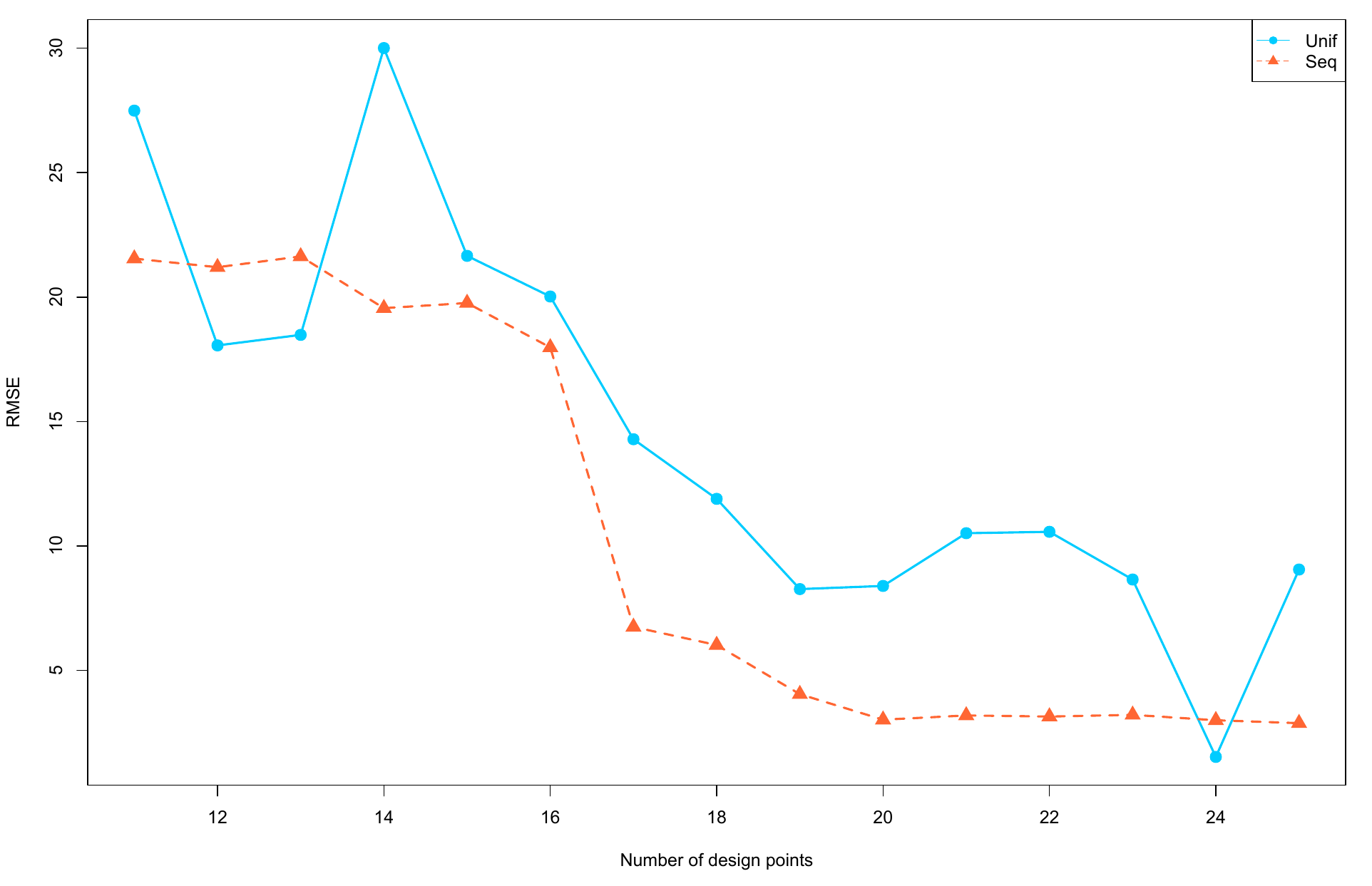}
	\caption{The RMSE of the uniform design and the gradient-based sequential design with different number of design points.}
	\label{fig:2}
\end{figure}

Figure \ref{fig:2} shows that when we utilize the information from each round of observations and incorporate new design points in regions with larger gradients, the RMSEs are lower than those obtained by the uniform design. This indicates that an efficient design should not only possess space-filling properties but also align as closely as possible with the distribution of gradient magnitudes. Moreover, in contrast to sequential designs, since uniform designs are regenerated in each round, using them to fit the model can lead to significant instability. Consequently, sequential designs are demonstrably superior to single-stage space-filling designs for this purpose. It is important to note that augmenting an existing design with new points does not always reduce the RMSE, as the new points may cause the model's predictive trends in certain regions to deviate further from the true response model.

\subsubsection{Specific expression of the gradient-based criterion}
Building on our prior discussion, the design points we select ought to maximize $\|\nabla f(\boldsymbol{x})\|$ while maintaining a sufficient distance from the existing design points. Let $\boldsymbol{X}$ denote the current design and $\boldsymbol{x}^*=\arg\min_{y\in\boldsymbol{X}}d(\boldsymbol{x},\boldsymbol{y})$. Therefore, an intuitive choice is 
\begin{equation}
	\label{eq:3}
	\phi(\boldsymbol{x})=\|\nabla f(\boldsymbol{x})\|\times d(\boldsymbol{x},\boldsymbol{x}^*),
\end{equation} 
where $d(\boldsymbol{x},\boldsymbol{x}^*)$ is the $L_2$-distance between $\boldsymbol{x}$ and $\boldsymbol{x}^*$. However, the true value of \( \|\nabla f(\boldsymbol{x})\| \) is unknown. Fortunately, we can compute the expectation of \( \|\nabla f(\boldsymbol{x})\|^2 \) based on the fitted Kriging model according to the following lemma.
\begin{lemma}[\citealp{Erickson:2021}]
	For a Gaussian process with correlation function $k(\boldsymbol{x},\boldsymbol{y})$ and variance $\tau^2$, we have
	\begin{equation*}
		E\left(\|\nabla f(\boldsymbol{x})\|^2\right)=\|\nabla\hat{y}_{\boldsymbol{X}}(\boldsymbol{x})\|^2+\tau^2tr(\nabla^2k(\boldsymbol{x},\boldsymbol{x})-\nabla r_{\boldsymbol{X}}^T(\boldsymbol{x})\boldsymbol{K}^{-1}(\boldsymbol{X})\nabla r_{\boldsymbol{X}}(\boldsymbol{x})),
	\end{equation*}
	where $\nabla\hat{y}_{\boldsymbol{X}}(\boldsymbol{x})=\nabla r^T_{\boldsymbol{X}}(\boldsymbol{x})\boldsymbol{K}^{-1}(\boldsymbol{X})\boldsymbol{Y}_{\boldsymbol{X}}$.
\end{lemma}

We may now substitute \( \sqrt{E\left(\|\nabla f(\boldsymbol{x})\|^2\right)} \) into \( \|\nabla f(\boldsymbol{x})\| \) in Equation \eqref{eq:3} to calculate \(\phi(\boldsymbol{x})\). It should be noted that there is an additional consideration. If a candidate design point exhibits a large gradient magnitude and is distant from existing design points, yet the variation in response values within its local neighborhood is negligible, such a point does not require excessive attention. This is due to the fact that points in this region exert minimal influence on RMSE. Accordingly, an additional term should be incorporated into \( \phi(\boldsymbol{x}) \) to quantify the local variability of \( f(\boldsymbol{x}) \) in the vicinity of \( \boldsymbol{x} \). We adopt $|\hat{y}_{\boldsymbol{X}}(\boldsymbol{x})-f(\boldsymbol{x}^*)|$ to denote the magnitude of response value variability in the neighborhood of \( \boldsymbol{x} \). A small value of this term indicates that $f(\boldsymbol{x})$ undergoes minimal variation near $\boldsymbol{x}$. Taking into account the previous discussions, we ultimately choose 
\begin{equation}
	\label{eq:4}
	\phi_{gra}(\boldsymbol{x})=\sqrt{E\left(\|\nabla f(\boldsymbol{x})\|^2\right)}\times d(\boldsymbol{x},\boldsymbol{x}^*)+|\hat{y}_{\boldsymbol{X}}(\boldsymbol{x})-f(\boldsymbol{x}^*)|.
\end{equation}
Furthermore, the theorem presented below establishes that the prediction error of the Kriging model can be approximately bounded by $\phi_{\text{gra}}(\boldsymbol{x})$ as defined in Equation \eqref{eq:4}.
\begin{theorem}
	\label{the:1}
	Let $f(\boldsymbol{x})$ be the smooth true model and $\hat{y}_{\boldsymbol{X}}(\boldsymbol{x})$ be the predicted value of the Kriging model at point \( \boldsymbol{x} \) based on \( \boldsymbol{X} \). Then for any $\boldsymbol{x}\in\mathcal{X}$, we have
	\begin{equation}
		\label{eq:5}
		|f(\boldsymbol{x})-\hat{y}_{\boldsymbol{X}}(\boldsymbol{x})|\leq\|\nabla f(\boldsymbol{t})\| \times d(\boldsymbol{x},\boldsymbol{x}^*)+|\hat{y}_{\boldsymbol{X}}(\boldsymbol{x})-f(\boldsymbol{x}^*)|,
	\end{equation}
	where $\boldsymbol{x}^*=\arg\min_{\boldsymbol{y}\in\boldsymbol{X}}d(\boldsymbol{x},\boldsymbol{y})$ and $\boldsymbol{t}\in\{\alpha\boldsymbol{x}+(1-\alpha)\boldsymbol{x}^*\,|\,0\leq\alpha\leq1\}$.
\end{theorem}
The right-hand side of Equation \eqref{eq:5} in Theorem \ref{the:1} has a similar expression to \(\phi_{gra}(\boldsymbol{x})\) in Equation \eqref{eq:4}. When \( d(\boldsymbol{x},\boldsymbol{x}^*) \) is small, we can consider \( \|\nabla f(\boldsymbol{t})\| \) and \( \sqrt{E\left(\|\nabla f(\boldsymbol{x})\|^2\right)} \) to be approximately equal. From this perspective, the point selected based on \(\phi_{gra}(\boldsymbol{x})\) each time is the one where the upper bound of the error $|f(\boldsymbol{x})-\hat{y}_{\boldsymbol{X}}(\boldsymbol{x})|$ is the largest. In this way, we gradually reduce the error to achieve the goal of global fitting.

Figure \ref{fig:3} presents the addition order of each design point in \( \boldsymbol{D}_2 \) from Figure \ref{fig:1}, which is the sequential design constructed via the framework detailed in Algorithm \ref{alg:1} by employing \(\phi_{gra}\) as the selection criterion with $10$ initial design points. Figure \ref{fig:3} clearly shows that the new design points are incrementally added in a manner that aligns with the gradient magnitude of the response surface, which effectively reduces the average fitting error. For instance, the gradient magnitude is relatively large in the bottom-left region, resulting in a denser distribution of augmented points; conversely, the gradient magnitude is smaller in the central region, leading to a sparser distribution of added points.
\begin{figure}[htbp]
	\centering
	\includegraphics[width=0.8\textwidth]{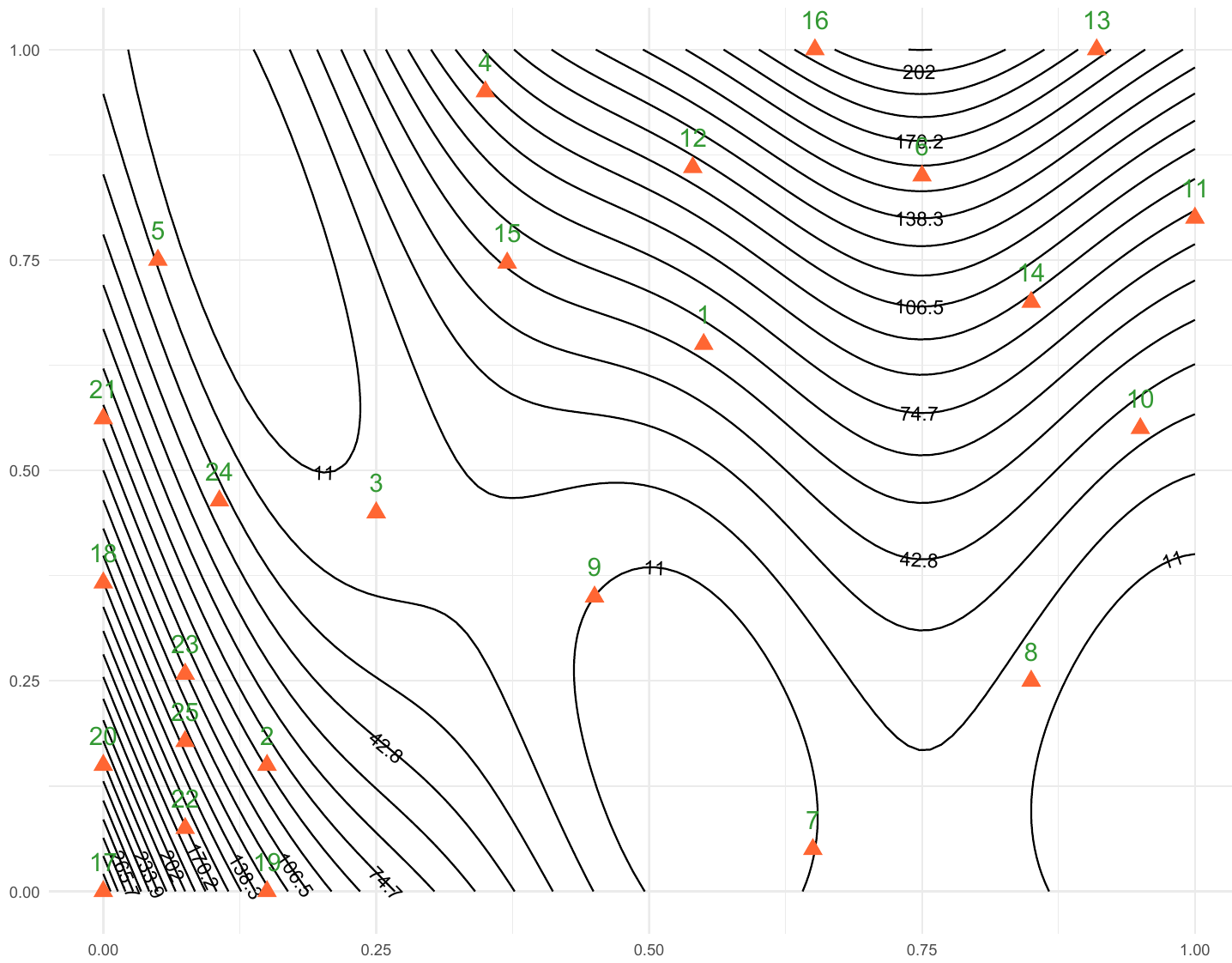}
	\caption{The order of design points in $\boldsymbol{D}_2$.}
	\label{fig:3}
\end{figure}

\subsection{Variance-based criterion}
When our objective is to minimize the maximum fitting error, we may approach this problem from an alternative perspective. \cite{Sacks:1988} proposed a criterion that selects the point minimizing the maximum prediction variance $\widehat{s_{\boldsymbol{X}}^2}(\boldsymbol{x})$ for all $\boldsymbol{x}\in\mathcal{X}$. However, relying solely on the prediction variance of design point estimates is insufficient to fully characterize the overall estimation error. The following theorem establishes that the estimation error at each point can be regulated not only by $\widehat{s_{\boldsymbol{X}}^2}(\boldsymbol{x})$ but also by the variance of the gradient estimate.

\begin{theorem}
	\label{the:2}
	Let $\mathcal{X}\subseteq\mathbb{R}$ be open, $\mathcal{N}_k(\mathcal{X})$ be the native Hilbert function space corresponding to the symmetric positive
	definite kernel $k(\boldsymbol{x},\boldsymbol{y}) \subseteq C^{2m}(\mathcal{X}, \mathcal{X}):\mathcal{X}\times\mathcal{X}\rightarrow\mathbb{R}$ and $f$ be the smooth true function. Denote the \(i\)-th diagonal element of $\nabla^2k(\boldsymbol{x},\boldsymbol{x})-\nabla r_{\boldsymbol{X}}^T(\boldsymbol{x})\boldsymbol{K}^{-1}(\boldsymbol{X})\nabla r_{\boldsymbol{X}}(\boldsymbol{x})$ as \(g_i(\boldsymbol{x})\). Then we have
	\begin{enumerate}[1.]
		\item $|f(\boldsymbol{x})-\hat{y}_{\boldsymbol{X}}(\boldsymbol{x})|\leq\sqrt{k(\boldsymbol{x},\boldsymbol{x})-\boldsymbol{r}_{\boldsymbol{X}}^T(\boldsymbol{x})\boldsymbol{K}^{-1}(\boldsymbol{X})\boldsymbol{r}_{\boldsymbol{X}}(\boldsymbol{x})}\times|f|_{\mathcal{N}_k(\mathcal{X})} = \sqrt{\widehat{s_{\boldsymbol{X}}^2}(\boldsymbol{x}) / \tau^2}|f|_{\mathcal{N}_k(\mathcal{X})}$,
		\item $|f(\boldsymbol{x})-\hat{y}_{\boldsymbol{X}}(\boldsymbol{x})|\leq\sqrt{\sum_{i=1}^mg_i(\boldsymbol{t})}\times d(\boldsymbol{x},\boldsymbol{x}^*)\times|f|_{\mathcal{N}_k(\mathcal{X})}$,
	\end{enumerate}
	where $\boldsymbol{t}\in\{\alpha\boldsymbol{x}+(1-\alpha)\boldsymbol{x}^*\,|\,0\leq\alpha\leq1\}$ and $|f|_{\mathcal{N}_k(\mathcal{X})}$ is the norm of function \( f \) in the native space $\mathcal{N}_k(\mathcal{X})$.
\end{theorem}
The expression in Theorem \ref{the:2} can be decomposed into two components: one component is governed by the variance of design point estimation or the variance of design point gradient estimation, which is the controllable part; the other part is determined by the true function \( f \) and the correlation function in the Gaussian process, which can be regarded as systematic error and is beyond our control. For the second inequality in Theorem \ref{the:2}, we can substitute $\boldsymbol{t}$ with $\boldsymbol{x}$ approximately. Based on the previous discussion, we can present the specific expression for the variance-based criterion as
$$\phi_{var}(\boldsymbol{x})=\min\left(\sqrt{k(\boldsymbol{x},\boldsymbol{x})-\boldsymbol{r}_{\boldsymbol{X}}^T(\boldsymbol{x})\boldsymbol{K}^{-1}(\boldsymbol{X})\boldsymbol{r}_{\boldsymbol{X}}(\boldsymbol{x})},\sqrt{\sum_{i=1}^mg_i(\boldsymbol{x})}\times d(\boldsymbol{x},\boldsymbol{x}^*)\right).$$
In each iteration, we gradually control the prediction error by finding the point $\boldsymbol{x}$ that maximizes $\phi_{var}$. In the next section, we will present an algorithmic framework for constructing batch sequential designs, which is applicable to any single-point sequential criterion.

\section{Cluster-based top-$b$ sequential design}\label{sec:4}
Earlier in this study, we elaborated on the detailed procedures and associated criteria for one-point sequential design. However, in practical scenarios---for instance, when experimental observations require substantial resource investment, yet the cost differential between observing a single point and multiple points is negligible---adopting a one-point-at-a-time approach to select and observe design points proves to be highly time-consuming and resource-prohibitive. When the batch \( b \) in Algorithm \ref{alg:1} is greater than $1$, we cannot simply select the \( b \) points that maximize the value of \( \phi(\boldsymbol{x}) \). This is because \( \phi(\boldsymbol{x}) \) is a continuous function, and these \( b \) points would inevitably cluster around the point that maximizes \( \phi \). 

In fact, in each iteration when we search for $\boldsymbol{x}_{n_0+i}=\arg\max_{\boldsymbol{x}\in\mathcal{X}}\phi(\boldsymbol{x})$ in a one-point sequential design, we can also obtain more information. Due to the continuity of \( \phi(\boldsymbol{x}) \), the points that maximize it tend to cluster together. In this situation, we can search for another local maximum, which also has a relatively large \( \phi \) value and is not in the vicinity of \( \boldsymbol{x}_i \). Then the point selected in the next iteration may well be the second local maximum of $\phi(\boldsymbol{x})$ in the current iteration, since the latter is likely the second most informative point besides the point that maximizes $\phi(\boldsymbol{x})$. We will take the $\phi_{gra}$ criterion as an example to illustrate this via the following experiments. Similar conclusions are valid for most of the remaining criteria, and thus they are not presented herein.

Let $f(x_1,x_2)$ be the function denoted in Equation \eqref{eq:2}. Assume that $n_0=10$ and $\boldsymbol{X}_0$ is a uniform design. To reduce computation time, we can discretize \( \mathcal{X}=[0,1]^2 \) using a uniform design. Let \( \boldsymbol{D}_{all} \) be a uniform design with \( n_{all} \) rows, which is also generated by the ``GenUD'' function from the ``UniDOE'' R package. In this way, we only need to compute \( \phi_{gra}(\boldsymbol{x}) \) on \( \boldsymbol{D}_{all} \) and select \( \boldsymbol{B}_i \) from \( \boldsymbol{D}_{all} \). Here we set $n_{all}=1000$.
\begin{figure}[htbp]
	\centering
	\begin{subfigure}[t]{0.32\textwidth}
		\centering
		\includegraphics[width=\textwidth]{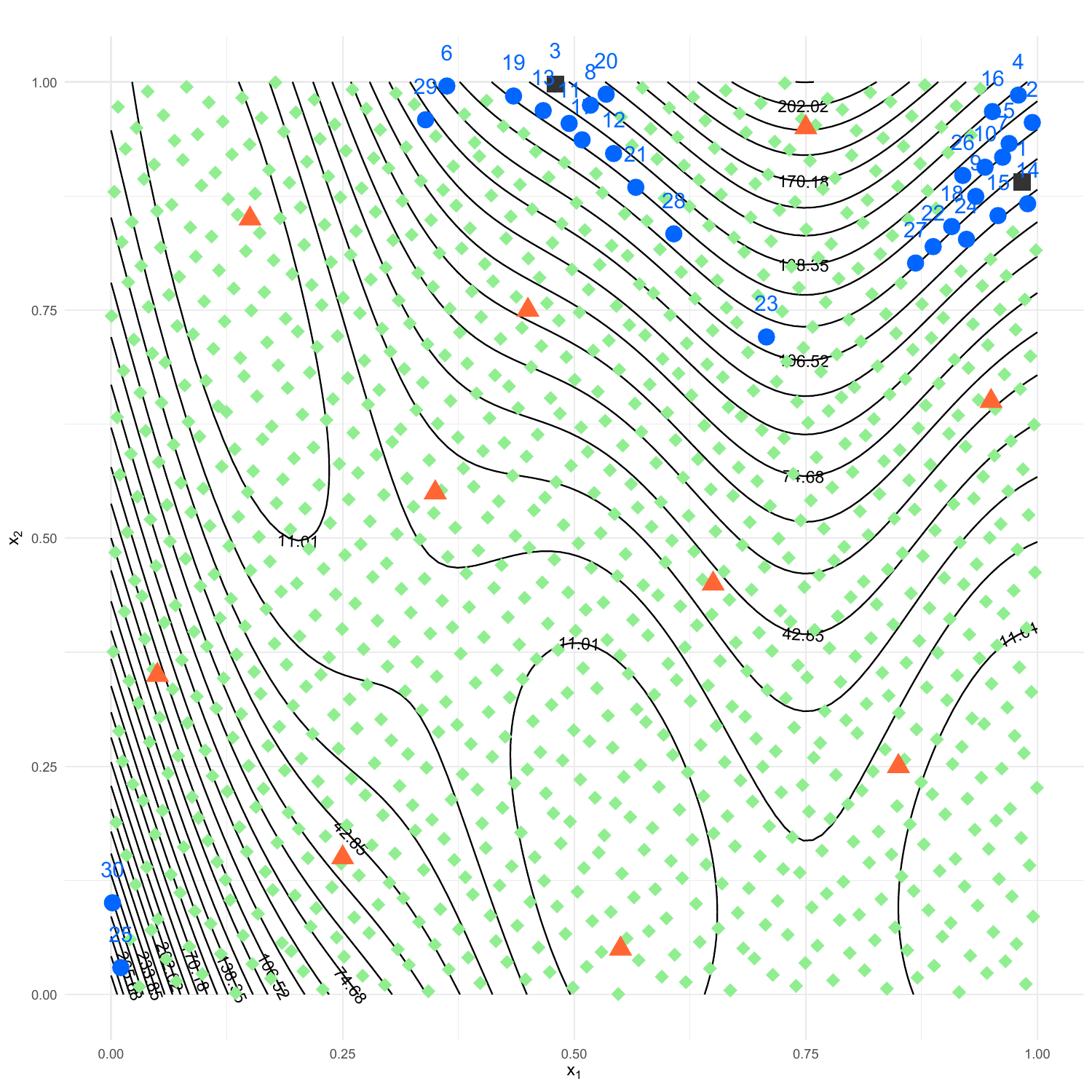}
		\caption{One-point sequential design with $10$ points.}
		\label{fig:fig41}
	\end{subfigure}
	\hfill
	\begin{subfigure}[t]{0.32\textwidth}
		\centering
		\includegraphics[width=\textwidth]{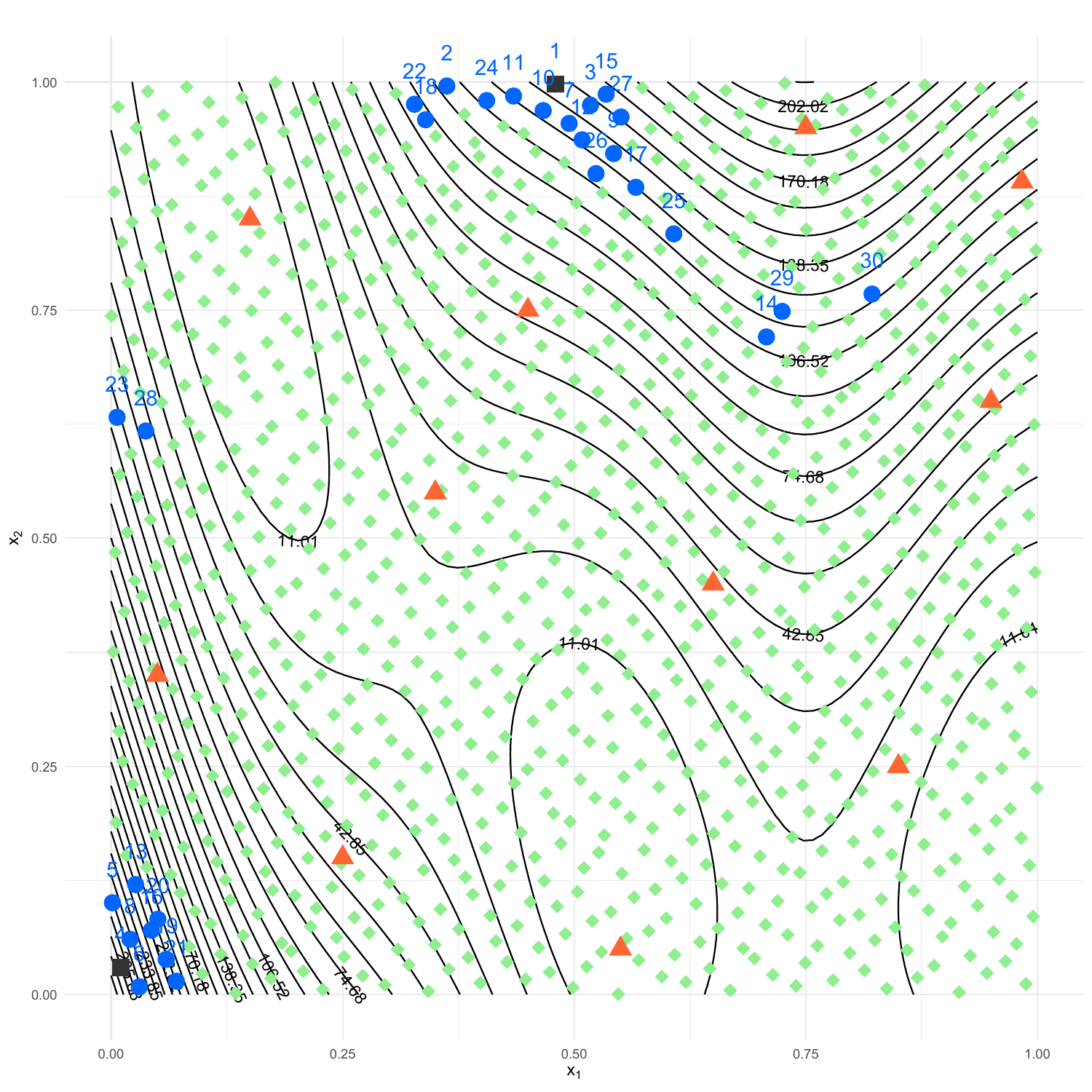}
		\caption{One-point sequential design with $11$ points.}
		\label{fig:fig42}
	\end{subfigure}
	\hfill
	\begin{subfigure}[t]{0.32\textwidth}
		\centering
		\includegraphics[width=\textwidth]{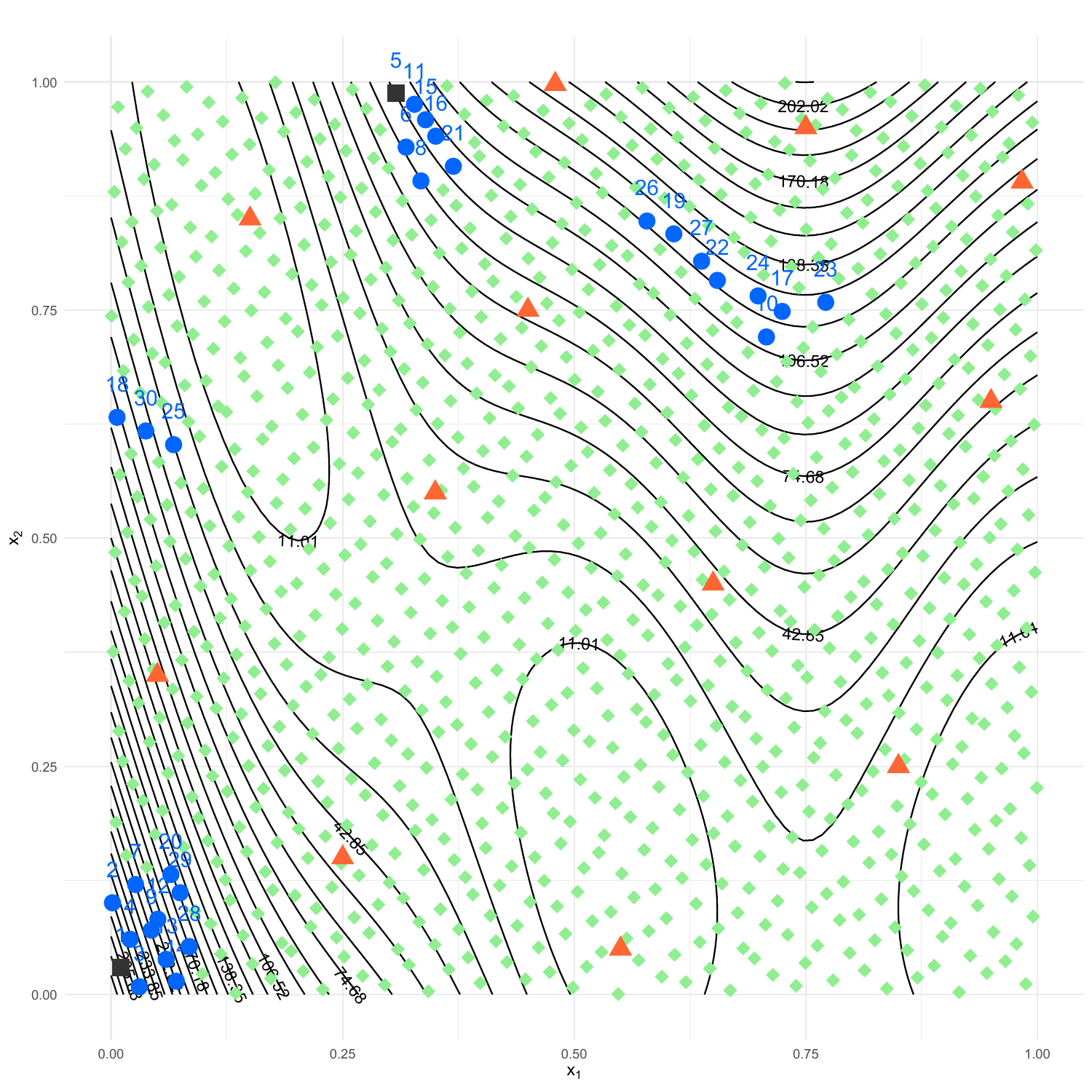}
		\caption{One-point sequential design with $12$ points.}
		\label{fig:fig43}
	\end{subfigure}
	\caption{The one-point sequential design and the order of values of \( \phi(\boldsymbol{x}) \) on \( \boldsymbol{D}_{all} \).}
	\label{fig:4}
\end{figure}

Figure \ref{fig:4} shows the points of one-point sequential design $\boldsymbol{X}_i$ and the order of the $\phi_{gra}$ values for the points in $\boldsymbol{D}_{all}$. The three subplots in Figure \ref{fig:4} correspond to the cases where the number of rows in \( \boldsymbol{X}_i \) is $10$, $11$, and $12$, respectively. In Figure \ref{fig:4}, the points of the current design \( \boldsymbol{X}_i \) are represented by red triangles, the points in \( \boldsymbol{D}_{all}\) are represented by green diamonds, and the blue circles represent the top $30$ points in \( \boldsymbol{D}_{all} \) that maximize \( \phi_{gra}(\boldsymbol{x}) \). The number above each point indicates its order. The black squares represent the points selected by clustering the points in \( \boldsymbol{D}_{all} \) based on the value of \( \phi_{gra}(\boldsymbol{x}) \) and choosing the points with the largest \( \phi_{gra} \) values in the first two clusters. In each iteration, the point labeled $1$ is selected as $\boldsymbol{B}_i$.

From Figure \ref{fig:4}, we can observe that, as noted earlier, the points maximizing $\phi_{gra}$ tend to exhibit a clustering pattern. In Figure \ref{fig:4}(a), the top two clusters maximizing $\phi_{gra}$ are represented by the point labeled $1$ and the point labeled $3$, respectively. For the one-point sequential design, the point labeled $1$ would typically be selected as the observation candidate. However, it becomes evident that after observing the point labeled $1$ and recalculating $\phi_{gra}$ with the updated data, the point labeled $3$ in Figure \ref{fig:4}(a) becomes the new maximizer of $\phi_{gra}(\boldsymbol{x})$. On this basis, we can directly select both the point labeled $1$ and the point labeled $3$ for simultaneous observation in the first iteration. This batch design selection strategy not only curtails the cost incurred by additional iterative observations but also yields outcomes that are comparable to those of the one-point sequential design. This is the main idea of the cluster-based top-$b$ sequential design.

In the cluster-based top-$b$ sequential design, the manner of partitioning the clusters stands as a pivotal element. Algorithm \ref{alg:2} shows the specific steps for partitioning the clusters of $\boldsymbol{D}_{all}$. First, the points in $\boldsymbol{D}_{all}$ are arranged in descending order of their \(\phi\) values, denoted as \(\boldsymbol{x}_{(1)}, \ldots, \boldsymbol{x}_{(n_{all})}\), with \(\boldsymbol{x}_{(1)}\) assigned to the first cluster. Then, for each point from \(\boldsymbol{x}_{(2)}\) to \(\boldsymbol{x}_{(n_{all})}\), it is determined whether it belongs to any existing cluster. If it is found to be part of an existing cluster, it is incorporated into that cluster; otherwise, it is assigned to a new cluster. To determine whether a point $\boldsymbol{x}$ belongs to an existing cluster, a set of rules is applied. If the cluster contains only one point, it is checked whether the cube formed by this point and \(\boldsymbol{x}\) encompasses more than \(\alpha\) points in $\boldsymbol{D}_{all}$; if so, they are considered not to belong to the same cluster. If the cluster contains more than one point, two calculations are performed: the average distance of each point in the cluster to the cluster centroid, as well as the distance between point \(\boldsymbol{x}\) and the centroid. If the distance from point \(\boldsymbol{x}\) to the centroid surpasses \(\beta\) times the average distance, the points are considered to be from different clusters. Once the number of clusters reaches $b$, the clustering process can be terminated, and the first point of each cluster is marked as the point for observation. It is worth noting that when the number of clusters is less than \( b \), we can decrease the value of $\alpha$ and repeat the clustering operation until the number of clusters attains \( b \). The points represented by the black squares in Figure \ref{fig:4} are those obtained by setting \(\alpha = 10\), \(\beta = 5\), and \(b = 2\) in Algorithm \ref{alg:2}.
\begin{algorithm}[htbp]
	\caption{The selection of $\boldsymbol{B}_i$ in top-$b$ sequential design}
	\label{alg:2}
	\begin{algorithmic}[1]
		\State \textbf{Input:} $\boldsymbol{D}_{all}$, $\phi(\boldsymbol{x})$, $b$, $\alpha$, $\beta$
		\State Sort the points in $\boldsymbol{D}_{all}$ by their \(\phi\) values in descending order and denoted them as \(\boldsymbol{x}_{(1)}, \ldots, \boldsymbol{x}_{(n_{all})}\)
		\State $Cluster=\{\{\boldsymbol{x}_{(1)}\}\}$
		\For{$i = 2$ \textbf{to} $n_{all}$}
		\If{the number of clusters in $Cluster$ is $b$}
		\State \textbf{break}
		\EndIf
		\For{each cluster $C$ in $Cluster$}
		\If{the number of points in $C$ is greater than $1$}
		\State $\bar{d}\leftarrow$ the average distance of each point in $C$ to the cluster centroid
		\State $d\leftarrow$ the distance between point \(\boldsymbol{x}\) and the centroid
		\If{$d<\beta\bar{d}$}
		\State Assign \( \boldsymbol{x}_{(i)} \) to cluster \( C \)
		\State \textbf{break}
		\EndIf
		\EndIf
		\If{the number of points in $C$ is $1$}
		\State Denote the point in $C$ as $\boldsymbol{x}_{C}$
		\State $num\leftarrow$ the number of points in the cube formed by $\boldsymbol{x}_{C}$ and \(\boldsymbol{x}\) 
		\If{$num \leq \alpha$}
		\State Assign \( \boldsymbol{x}_{(i)} \) to cluster \( C \)
		\State \textbf{break}
		\EndIf
		\EndIf
		\EndFor
		\State Add \( \boldsymbol{x}_{(i)} \) as a new cluster to \( Cluster \).
		\EndFor
		\State \textbf{Output: } the first point of each cluster in $Cluster$
	\end{algorithmic}
\end{algorithm}

Figure \ref{fig:5} shows the cluster-based top-$b$ sequential design generated based on the initial design in Figure \ref{fig:4}(a) with \( b = 3 \), \( \alpha = 10 \), \( \beta = 5 \) and a total of $7$ iterations under the $\phi_{gra}$ criterion. The label assigned to each point indicates the iteration in which it was added, with the label ``0'' representing the initial design. It can be observed from Figure \ref{fig:5} that the point distribution pattern of the batch sequential design generated by Algorithm \ref{alg:2} under the gradient-based criterion \(\phi_{gra}\) is highly similar to that of the single-point sequential design: points are densely distributed in the bottom-left region with large gradient magnitudes, while they are sparsely distributed in the central region with small gradient magnitudes, and the overall design maintains favorable space-filling properties. This finding demonstrates that the batch sequential design generated by Algorithm \ref{alg:2} can retain the core characteristics of the selected criterion to a considerable extent.
\begin{figure}[htbp]
	\centering
	\includegraphics[width=0.7\textwidth]{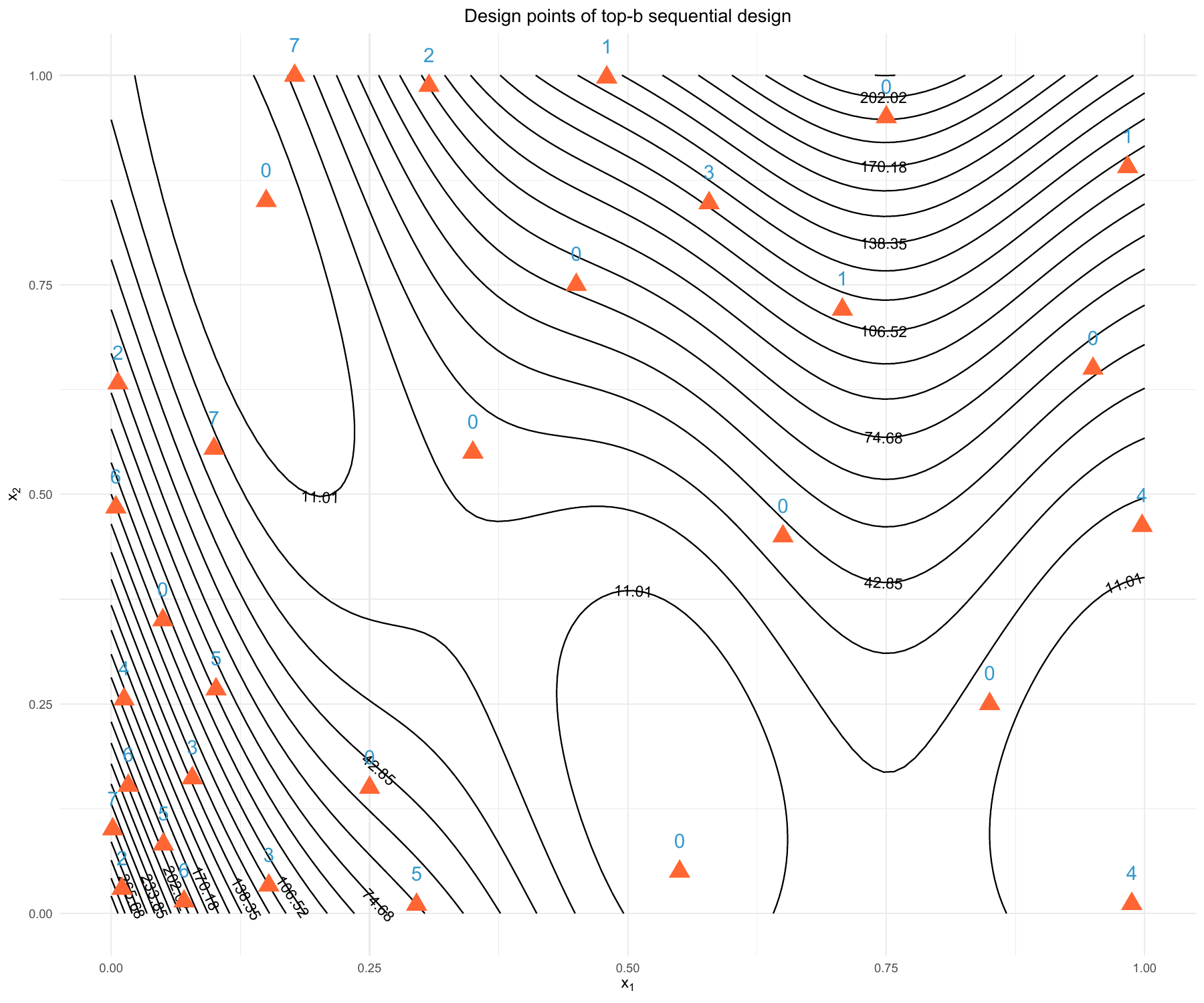}
	\caption{The order of points in a cluster-based top-$b$ sequential design.}
	\label{fig:5}
\end{figure}

\section{Simulations}
\label{sec:5}
In the present section, we systematically evaluate and compare the performance of sequential designs generated based on different criteria across a suite of test functions. To this end, we first provide a detailed introduction to the test functions employed in the subsequent analyses. Subsequently, we conduct a comparative assessment of the performance of one-point sequential designs and batch sequential designs, respectively.

\subsection{Test functions}
The first test function is the well-known Branin function. Its specific expression is presented in Equation \eqref{eq:2}, and we denote it as $f_1$. The second and the third test functions \citep{Mu:2016} are as follows,
$$
f_2(x_1,x_2)=\left(1 - \exp\left(-\frac{0.5}{x_2}\right)\right) \frac{2300x_1^3 + 1900x_1^2 + 2092x_1 + 60}{100x_1^3 + 500x_1^2 + 4x_1 + 20},
$$
and
$$f_3(x_1,x_2)=\log(2 + \sin(2\pi x_1))\cos(2\pi x_2^2),$$
respectively. The fourth test function is the Hartman function of dimension three. Its specific expression is 
$$f_4(x_1,x_2,x_3)=-\sum_{i = 1}^{4} c_i \exp\left\{ -\sum_{j = 1}^{3} \alpha_{ij}(x_j - p_{ij})^2 \right\},$$
where the coefficients \( c_i, \alpha_{ij} \), and \( p_{ij} \) for \( i = 1, \ldots, 4 \), \( j = 1,2,3 \) are respectively the $i$-th, $(i,j)$-th, and $(i,j)$-th elements of  
\[
\begin{pmatrix} 1 \\ 1.2 \\ 3 \\ 3.2 \end{pmatrix}, \quad \begin{pmatrix} 3 & 10 & 30 \\ 0.1 & 10 & 35 \\ 3 & 10 & 30 \\ 0.1 & 10 & 35 \end{pmatrix}, \text{ and } \begin{pmatrix} 0.3689 & 0.1170 & 0.2673 \\ 0.4699 & 0.4387 & 0.7470 \\ 0.1091 & 0.8732 & 0.5547 \\ 0.03815 & 0.5743 & 0.8828 \end{pmatrix}.
\] 
The fifth test function is the Ackley function of dimension five. Its specific expression is 
$$
f_5(\boldsymbol{x})=-20 \exp\left(-0.2 \sqrt{\frac{1}{5} \sum_{i=1}^{5} x_i^2}\right) - \exp\left(\frac{1}{5} \sum_{i=1}^{5} \cos(2\pi x_i)\right) + 20 + e.
$$

For the five test functions introduced above, we adopt those with fixed dimensions. In contrast, the subsequent two test functions we employ are functions with variable dimensions and distinctive geometric shapes. We denote the dimension of the subsequent test functions as \( m \). The sixth test function is the Zakharov function, which is a multidimensional flat-type function. Its specific form is
\[
f_6(\boldsymbol{x}) = \sum_{i=1}^{m} x_i^2 + \left( \sum_{i=1}^{m} 0.5 i x_i \right)^2 + \left( \sum_{i=1}^{m} 0.5 i x_i \right)^4.
\]  
The seventh test function is the Rosenbrock function, which is a multi-dimensional valley-shaped function. Its specific form is 
\[
f_7(\boldsymbol{x}) = \sum_{i=1}^{m-1} \left[ 100(x_{i + 1} - x_i^2)^2 + (x_i - 1)^2 \right].
\] 
The test functions $f_4$ to $f_7$ are all taken from \cite{Jamil:2013}. For convenience, all variables of the aforementioned test functions are constrained to the interval $[0, 1]$. Given that each of these test functions is continuous and infinitely differentiable, the surrogate models employed in the numerical simulations are all Kriging models with the correlation functions in the form of Equation \eqref{eq:6}.

\subsection{Comparison of one-point sequential designs}
In this section, we assess the performance of the one-point sequential designs generated based on seven different criteria: $\phi_{EI_0}$, $\phi_{EI_1}$, $\phi_{EI_2}$, $\phi_s$, $\phi_{MD}$, $\phi_{gra}$ and $\phi_{var}$. This evaluation is performed using the test functions presented in Section 5.1. For the two test functions \( f_6, f_7 \), we adopt dimensions of $3$, $5$, and $8$ for evaluation. For an $m$-dimensional test function, the initial design consists of $5m$ points, meaning the size of the initial design is proportional to the dimensionality of the test function. This is because all criteria are associated with the goodness-of-fit of the Gaussian process model. If the fitted Gaussian process model diverges significantly from the true model, these criteria may result in the selection of inappropriate design points, thus forming a vicious cycle. Consequently, as the dimensionality increases, more initial points are required to sufficiently explore the experimental domain prior to selecting subsequent points based on the proposed criteria. For all test functions, the number of added points is $20$ and the number of points in the test matrix is $10^5$. In each scenario, the test matrix is pre-determined and the performance of all one-point sequential designs is evaluated under this fixed test matrix.

To quantitatively assess the performance of a given design, we use it to fit a GP model with a correlation function in the form of Equation \eqref{eq:6}. Subsequently, we compute the RMSE and maximum absolute error (MAE) to quantify the predictive accuracy of the fitted GP model. For the computation of RMSE and MAE, we employ the ``randomLHS'' function from the ``lhs'' R package to generate a test matrix \( \boldsymbol{X}_{\text{test}} \), and the RMSE and MAE are computed on this matrix as follows, 
$$RMSE(\boldsymbol{X})=\sqrt{\frac{\sum_{\boldsymbol{x}\in\boldsymbol{X}_{test}}\left(f(\boldsymbol{x})-\hat{y}_{\boldsymbol{X}}(\boldsymbol{x})\right)^2}{n}}$$
and
$$MAE(\boldsymbol{X})=\max_{\boldsymbol{x}\in\boldsymbol{X}_{test}}|f(\boldsymbol{x})-\hat{y}_{\boldsymbol{X}}(\boldsymbol{x})|,$$
where $n$ denotes the number of points in the test matrix $\boldsymbol{X}_{\text{test}}$.

Tables \ref{tab:1} and \ref{tab:2} present the MAE and RMSE of different one-point sequential designs under different test functions, respectively. The bolded numbers in each row denote the design that exhibits the best performance under the corresponding scenario. As we can see from the tables, the best-performing design in most cases is one of the two sequential designs proposed in this paper. Additionally, the sequential designs derived from the variance-based criterion demonstrate superior performance in terms of MAE, while the sequential designs generated by gradient-based criterion perform better in terms of RMSE. Although there is no theoretical justification for the outperformance of designs generated by gradient-based criteria in terms of RMSE, an analysis of the underlying principles of the two criteria provides some insights. Specifically, the gradient-based criterion aims to add points uniformly based on the magnitude of gradients, whereas variance-based criterion theoretically selects points with the largest upper bound of absolute error according to Theorem \ref{the:2}. This discrepancy in objectives implies that designs generated by gradient-based criterion may yield better results for RMSE, while those generated based on the variance-based criterion are more effective in minimizing MAE.

\begin{table}[h!]
	\hfill
	\renewcommand{\arraystretch}{1.2}
	\caption{The MAE of different one-point sequential designs for different test functions.}
	\label{tab:1}
	\centering
	\normalsize
	\setlength{\tabcolsep}{4.5pt}
	\begin{tabular}{ccccccccc}
		\hline
		Test function&   $m$&   $\phi_{EI_0}$&   $\phi_{EI_1}$&   $\phi_{EI_2}$&   $\phi_{s}$&   $\phi_{MD}$&   $\phi_{gra}$&   $\phi_{var}$\\ \hline
		$f_1$&	$2$&	$49.61$&	$9.85$&	$3.07$&	$3.79$&	$21.30$&	$26.02$&	$\mathbf{2.05}$\\
		$f_2$&	$2$&	$2.41$&	$0.99$&	$0.82$&	$1.24$&	$3.26$&	$\mathbf{0.75}$&	$0.97$\\
		$f_3$&	$2$&	$0.84$&	$0.54$&	$0.26$&	$0.25$&	$0.37$&	$0.99$&	$\mathbf{0.21}$\\
		$f_4$&	$3$&	$1.00$&	$1.17$&	$1.27$&	$1.27$&	$0.96$&	$0.80$&	$\mathbf{0.80}$\\
		$f_5$&	$5$&	$14.49$&	$13.94$&	$14.05$&	$14.65$&	$13.91$&	$14.05$&	$\mathbf{13.69}$\\
		$f_6$&	$3$&	$1.90$&	$2.09$&	$1.15$&	$0.35$&	$1.17$&	$2.60$&	$\mathbf{0.28}$\\
		$f_6$&	$5$&	$235.12$&	$298.32$&	$\mathbf{112.92}$&	$213.95$&	$426.67$&	$228.48$&	$186.25$\\
		$f_6$&	$8$&	$14565$&	$14359$&	$10203$&	$25987$&	$25124$&	$\mathbf{9478}$&	$11158$\\
		$f_7$&	$3$&	$6.92$&	$2.95$&	$2.21$&	$1.44$&	$6.86$&	$6.23$&	$\mathbf{1.41}$\\
		$f_7$&	$5$&	$91.53$&	$163.05$&	$153.94$&	$61.67$&	$147.37$&	$104.37$&	$\mathbf{58.59}$\\
		$f_7$&	$8$&	$235.79$&	$248.43$&	$204.63$&	$179.14$&	$278.62$&	$255.74$&	$\mathbf{177.21}$\\
		\hline
	\end{tabular}
\end{table} 

\begin{table}[h!]
	\hfill
	\renewcommand{\arraystretch}{1.2}
	\caption{The RMSE of different one-point sequential designs for different test functions.}
	\label{tab:2}
	\centering
	\normalsize
	\setlength{\tabcolsep}{3.5pt}
	\begin{tabular}{ccccccccc}
		\hline
		Test function&   $m$&   $\phi_{EI_0}$&   $\phi_{EI_1}$&   $\phi_{EI_2}$&   $\phi_{s}$&   $\phi_{MD}$&   $\phi_{gra}$&   $\phi_{var}$\\
		\hline
		$f_1$&   $2$&   $4.49$&   $0.72$&   $0.72$&   $1.11$&   $1.70$&   $\mathbf{0.52}$&   $0.65$\\
		$f_2$&   $2$&   $0.43$&   $0.26$&   $0.25$&   $0.35$&   $0.30$&   $\mathbf{0.21}$&   $0.23$\\
		$f_3$&   $2$&   $0.11$&   $0.09$&   $0.07$&   $0.07$&   $0.08$&   $\mathbf{0.06}$&   $0.06$\\
		$f_4$&   $3$&   $0.20$&   $0.21$&   $0.23$&   $0.26$&   $0.20$&   $\mathbf{0.17}$&   $0.31$\\
		$f_5$&   $5$&   $0.91$&   $0.81$&   $0.82$&   $0.77$&   $0.81$&   $0.81$&   $\mathbf{0.77}$\\
		$f_6$&   $3$&   $0.18$&   $0.15$&   $0.10$&   $0.08$&   $0.15$&   $0.22$&   $\mathbf{0.06}$\\
		$f_6$&	$5$&	$24.05$&	$25.32$&	$25.06$&	$28.83$&	$24.92$&	$\mathbf{23.11}$&	$35.49$\\
		$f_6$&   $8$&   $1637.46$&   $1611.18$&   $\mathbf{1162.48}$&   $2103.61$&   $1349.59$&   $1207.69$&   $1487.18$\\
		$f_7$&   $3$&   $0.76$&   $0.42$&   $0.30$&   $0.28$&   $0.59$&   $\mathbf{0.24}$&   $0.30$\\
		$f_7$&   $5$&   $16.06$&   $19.92$&   $15.19$&   $\mathbf{13.18}$&   $17.71$&   $15.68$&   $13.25$\\
		$f_7$&   $8$&   $37.11$&   $42.03$&   $35.40$&   $30.97$&   $47.88$&   $\mathbf{29.61}$&   $30.34$\\
		\hline
	\end{tabular}
\end{table} 

\subsection{Comparison of batch sequential design}
For batch sequential design, we adopt the sequential design generated based on MD as the benchmark. Then we compare the performance of the cluster-based top-$b$ designs obtained using the $\phi_{EI_0}$, $\phi_{EI_1}$, $\phi_{EI_2}$, $\phi_s$, $\phi_{gra}$ and $\phi_{var}$ criteria under different test functions. The test functions, their corresponding dimensions, the number of initial points, and the number of points in the test matrix are consistent with those in the one-point sequential design. For various scenarios, we separately present the performance when \( b = 2, 3, 5 \). 
In all cases, we set the number of iterations to $10$ and specify the parameters as \( \alpha = 15 \) and \( \beta = 5 \).

Tables \ref{tab:3.1}--\ref{tab:4.2} show the MAE and RMSE of different batch sequential designs across various test functions. From these tables, it can be observed that the performance of the cluster-based top-$b$ sequential designs generated based on gradient-based or variance-based criteria outperforms that of other comparative designs in the majority of test cases. Furthermore, designs based on the variance-based criterion exhibit superior performance in terms of MAE, whereas those based on the gradient-based criterion demonstrate better performance in terms of RMSE. This finding is consistent with the simulation results of the one-point sequential designs reported in Tables \ref{tab:1} and \ref{tab:2}. In addition, within the framework of Algorithm \ref{alg:2}, even when other criteria are employed, the resulting batch sequential designs outperform the sequential uniform design, which indicates the effectiveness of Algorithm \ref{alg:2}. It is worth noting that when $b = 2$, the total number of points in the batch sequential design is the same as that of the one-point sequential designs presented in Tables \ref{tab:1} and \ref{tab:2}. Through comparison, we can find that, in some cases, the batch sequential design is even better than the corresponding one-point sequential design, while in other cases, their performances are nearly indistinguishable. This phenomenon may be attributed to the fact that selecting multiple points in each iteration mitigates the impact of inaccurate model fitting to a certain extent. It also demonstrates the potential of leveraging more information to select multiple points per iteration.  

\begin{table}[h!]
	\hfill
	\renewcommand{\arraystretch}{1.2}
	\caption{The MAE of different batch sequential designs for different test functions.}
	\label{tab:3.1}
	\centering
	\normalsize
	\begin{tabular}{cccccccccc}
		\hline
		Test function&	$m$&	$b$&	$\phi_{EI_0}$&	$\phi_{EI_1}$&	$\phi_{EI_2}$&	$\phi_{s}$&	$\phi_{MD}$&	$\phi_{gra}$&	$\phi_{var}$\\
		\hline
		$f_1$&	$2$&	$2$&	$4.73$&	$5.63$&	$2.56$&	$3.41$&	$11.41$&	$4.92$&	$\mathbf{1.64}$\\
		$f_1$&	$2$&	$3$&	$4.88$&	$1.94$&	$1.28$&	$0.77$&	$4.18$&	$4.49$&	$\mathbf{0.72}$\\
		$f_1$&	$2$&	$5$&	$5.39$&	$0.36$&	$0.30$&	$\mathbf{0.08}$&	$2.66$&	$8.00$&	$0.11$\\
		$f_2$&	$2$&	$2$&	$1.95$&	$1.08$&	$1.03$&	$0.85$&	$4.32$&	$0.90$&	$\mathbf{0.83}$\\
		$f_2$&	$2$&	$3$&	$1.33$&	$0.99$&	$0.42$&	$0.39$&	$1.39$&	$1.66$&	$\mathbf{0.38}$\\
		$f_2$&	$2$&	$5$&	$0.56$&	$0.50$&	$0.32$&	$0.37$&	$0.66$&	$0.69$&	$\mathbf{0.21}$\\
		$f_3$&	$2$&	$2$&	$0.62$&	$0.65$&	$0.30$&	$0.25$&	$1.09$&	$0.68$&	$\mathbf{0.21}$\\
		$f_3$&	$2$&	$3$&	$0.31$&	$0.15$&	$0.23$&	$\mathbf{0.09}$&	$0.25$&	$0.41$&	$0.11$\\
		$f_3$&	$2$&	$5$&	$0.18$&	$0.07$&	$0.10$&	$0.06$&	$0.31$&	$0.12$&	$\mathbf{0.05}$\\
		$f_4$&	$3$&	$2$&	$1.27$&	$1.19$&	$\mathbf{0.84}$&	$1.28$&	$1.11$&	$1.21$&	$1.49$\\
		$f_4$&	$3$&	$3$&	$1.20$&	$1.19$&	$0.87$&	$1.00$&	$0.90$&	$1.18$&	$\mathbf{0.84}$\\
		$f_4$&	$3$&	$5$&	$1.00$&	$0.94$&	$0.87$&	$\mathbf{0.54}$&	$0.82$&	$1.18$&	$1.01$\\
		$f_5$&	$5$&	$2$&	$13.30$&	$12.66$&	$12.75$&	$13.15$&	$12.78$&	$12.71$&	$\mathbf{11.15}$\\
		$f_5$&	$5$&	$3$&	$\mathbf{11.38}$&	$12.86$&	$12.78$&	$12.68$&	$12.68$&	$12.13$&	$12.89$\\
		$f_5$&	$5$&	$5$&	$12.61$&	$12.92$&	$12.68$&	$12.79$&	$12.55$&	$12.11$&	$\mathbf{11.87}$\\	
		\hline
	\end{tabular}
\end{table} 

\begin{table}[h!]
	\hfill
	\renewcommand{\arraystretch}{1.2}
	\caption{The MAE of different batch sequential designs for different test functions.}
	\label{tab:3.2}
	\normalsize
	\setlength{\tabcolsep}{4.5pt}
	\centering
	\begin{tabular}{cccccccccc}
		\hline
		Test function&	$m$&	$b$&	$\phi_{EI_0}$&	$\phi_{EI_1}$&	$\phi_{EI_2}$&	$\phi_{s}$&	$\phi_{MD}$&	$\phi_{gra}$&	$\phi_{var}$\\
		\hline
		$f_6$&	$3$&	$2$&	$2.98$&	$1.47$&	$1.55$&	$0.60$&	$1.98$&	$2.50$&	$\mathbf{0.53}$\\
		$f_6$&	$3$&	$3$&	$2.15$&	$0.78$&	$0.72$&	$0.14$&	$0.43$&	$2.56$&	$\mathbf{0.13}$\\
		$f_6$&	$3$&	$5$&	$0.78$&	$0.21$&	$0.18$&	$0.15$&	$1.61$&	$2.02$&	$\mathbf{0.07}$\\
		$f_6$&	$5$&	$2$&	$193.33$&	$169.64$&	$\mathbf{141.85}$&	$150.62$&	$478.56$&	$203.35$&	$385.72$\\
		$f_6$&	$5$&	$3$&	$118.38$&	$111.96$&	$157.04$&	$139.15$&	$162.42$&	$\mathbf{77.20}$&	$156.14$\\
		$f_6$&	$5$&	$5$&	$92.04$&	$55.83$&	$55.94$&	$70.80$&	$168.31$&	$131.43$&	$\mathbf{54.11}$\\
		$f_6$&	$8$&	$2$&	$\mathbf{6522}$&	$15147$&	$17628$&	$23153$&	$21489$&	$7312$&	$12154$\\
		$f_6$&	$8$&	$3$&	$14265$&	$9981$&	$\mathbf{8980}$&	$13134$&	$28304$&	$14096$&	$14422$\\
		$f_6$&	$8$&	$5$&	$7502$&	$14194$&	$7209$&	$7158$&	$12399$&	$14939$&	$\mathbf{6336}$\\
		$f_7$&	$3$&	$2$&	$4.45$&	$6.62$&	$4.38$&	$1.88$&	$4.56$&	$5.47$&	$\mathbf{1.69}$\\
		$f_7$&	$3$&	$3$&	$8.54$&	$3.82$&	$1.79$&	$0.70$&	$2.30$&	$6.66$&	$\mathbf{0.67}$\\
		$f_7$&	$3$&	$5$&	$2.17$&	$1.48$&	$1.81$&	$\mathbf{0.40}$&	$2.65$&	$2.47$&	$0.49$\\
		$f_7$&	$5$&	$2$&	$156.86$&	$108.34$&	$125.45$&	$66.38$&	$109.19$&	$148.38$&	$\mathbf{56.24}$\\
		$f_7$&	$5$&	$3$&	$111.60$&	$108.38$&	$65.05$&	$52.09$&	$119.98$&	$137.97$&	$\mathbf{49.70}$\\
		$f_7$&	$5$&	$5$&	$37.21$&	$27.49$&	$21.39$&	$18.68$&	$30.39$&	$43.75$&	$\mathbf{13.88}$\\
		$f_7$&	$8$&	$2$&	$283.52$&	$178.16$&	$304.74$&	$295.58$&	$\mathbf{175.32}$&	$342.52$&	$296.69$\\
		$f_7$&	$8$&	$3$&	$172.04$&	$173.09$&	$143.10$&	$156.67$&	$183.60$&	$235.05$&	$\mathbf{139.05}$\\
		$f_7$&	$8$&	$5$&	$178.00$&	$130.25$&	$162.33$&	$122.13$&	$138.46$&	$115.40$&	$\mathbf{115.11}$\\	
		\hline
	\end{tabular}
\end{table} 

\begin{table}[h!]
	\hfill
	\renewcommand{\arraystretch}{1.2}
	\caption{The RMSE of different batch sequential designs for different test functions.}
	\label{tab:4.1}
	\normalsize
	\centering
	\begin{tabular}{cccccccccc}
		\hline
		Test function&	$m$&	$b$&	$\phi_{EI_0}$&	$\phi_{EI_1}$&	$\phi_{EI_2}$&	$\phi_{s}$&	$\phi_{MD}$&	$\phi_{gra}$&	$\phi_{var}$\\ \hline
		$f_1$&	$2$&	$2$&	$0.61$&	$1.26$&	$0.50$&	$7.92$&	$1.62$&	$\mathbf{0.40}$&	$0.43$\\
		$f_1$&	$2$&	$3$&	$0.45$&	$0.24$&	$0.18$&	$0.23$&	$0.42$&	$\mathbf{0.14}$&	$0.18$\\
		$f_1$&	$2$&	$5$&	$0.39$&	$0.03$&	$0.03$&	$0.02$&	$0.15$&	$0.07$&	$\mathbf{0.02}$\\
		$f_2$&	$2$&	$2$&	$0.33$&	$0.23$&	$0.25$&	$0.26$&	$0.28$&	$\mathbf{0.23}$&	$0.26$\\
		$f_2$&	$2$&	$3$&	$0.25$&	$0.18$&	$0.14$&	$0.12$&	$0.24$&	$0.33$&	$\mathbf{0.12}$\\
		$f_2$&	$2$&	$5$&	$0.15$&	$0.10$&	$0.06$&	$0.06$&	$0.08$&	$\mathbf{0.05}$&	$0.06$\\
		$f_3$&	$2$&	$2$&	$0.11$&	$0.09$&	$0.08$&	$0.08$&	$0.11$&	$\mathbf{0.06}$&	$0.07$\\
		$f_3$&	$2$&	$3$&	$0.083$&	$0.042$&	$0.028$&	$0.029$&	$0.032$&	$0.070$&	$\mathbf{0.028}$\\
		$f_3$&	$2$&	$5$&	$0.026$&	$0.011$&	$0.009$&	$0.014$&	$0.019$&	$0.029$&	$\mathbf{0.009}$\\
		$f_4$&	$3$&	$2$&	$0.21$&	$0.21$&	$0.26$&	$0.27$&	$0.21$&	$\mathbf{0.20}$&	$0.29$\\
		$f_4$&	$3$&	$3$&	$0.21$&	$0.23$&	$0.14$&	$0.20$&	$0.14$&	$\mathbf{0.12}$&	$0.15$\\
		$f_4$&	$3$&	$5$&	$0.20$&	$0.13$&	$0.16$&	$0.10$&	$0.11$&	$\mathbf{0.10}$&	$0.14$\\
		$f_5$&	$5$&	$2$&	$1.28$&	$0.83$&	$\mathbf{0.81}$&	$0.84$&	$0.81$&	$0.82$&	$0.84$\\
		$f_5$&	$5$&	$3$&	$0.90$&	$0.81$&	$0.81$&	$0.75$&	$0.77$&	$\mathbf{0.73}$&	$0.79$\\
		$f_5$&	$5$&	$5$&	$0.74$&	$0.73$&	$0.75$&	$0.66$&	$0.67$&	$\mathbf{0.63}$&	$0.69$\\
		\hline
	\end{tabular}
\end{table} 

\begin{table}[h!]
	\hfill
	\renewcommand{\arraystretch}{1.2}
	\caption{The RMSE of different batch sequential designs for different test functions.}
	\label{tab:4.2}
	\normalsize
	\setlength{\tabcolsep}{3pt}
	\centering
	\begin{tabular}{cccccccccc}
		\hline
		Test function&	$m$&	$b$&	$\phi_{EI_0}$&	$\phi_{EI_1}$&	$\phi_{EI_2}$&	$\phi_{s}$&	$\phi_{MD}$&	$\phi_{gra}$&	$\phi_{var}$\\ \hline
		$f_6$&	$3$&	$2$&	$0.36$&	$0.14$&	$0.14$&	$0.08$&	$0.12$&	$0.31$&	$\mathbf{0.08}$\\
		$f_6$&	$3$&	$3$&	$0.21$&	$0.05$&	$0.07$&	$\mathbf{0.03}$&	$0.05$&	$0.25$&	$0.03$\\
		$f_6$&	$3$&	$5$&	$0.08$&	$0.02$&	$0.02$&	$0.02$&	$0.06$&	$0.14$&	$\mathbf{0.01}$\\
		$f_6$&	$5$&	$2$&	$21.92$&	$15.24$&	$15.10$&	$\mathbf{14.34}$&	$24.04$&	$22.26$&	$20.84$\\
		$f_6$&	$5$&	$3$&	$15.88$&	$11.10$&	$12.10$&	$10.14$&	$10.66$&	$\mathbf{9.81}$&	$10.43$\\
		$f_6$&	$5$&	$5$&	$8.65$&	$\mathbf{3.28}$&	$3.44$&	$5.93$&	$5.87$&	$7.02$&	$4.61$\\
		$f_6$&	$8$&	$2$&	$1184.36$&	$1194.42$&	$1255.51$&	$1343.93$&	$1127.46$&	$\mathbf{967.89}$&	$1763.92$\\
		$f_6$&	$8$&	$3$&	$876.46$&	$1119.30$&	$1159.55$&	$1179.03$&	$1400.56$&	$1486.16$&	$1379.79$\\
		$f_6$&	$8$&	$5$&	$844.45$&	$1109.49$&	$890.40$&	$957.67$&	$997.87$&	$\mathbf{843.98}$&	$1146.62$\\
		$f_7$&	$3$&	$2$&	$0.51$&	$0.42$&	$0.31$&	$0.35$&	$0.48$&	$0.74$&	$\mathbf{0.28}$\\
		$f_7$&	$3$&	$3$&	$0.58$&	$0.28$&	$0.33$&	$0.14$&	$0.25$&	$0.44$&	$\mathbf{0.12}$\\
		$f_7$&	$3$&	$5$&	$0.18$&	$0.11$&	$0.12$&	$0.10$&	$0.23$&	$\mathbf{0.10}$&	$0.11$\\
		$f_7$&	$5$&	$2$&	$21.34$&	$11.53$&	$11.01$&	$11.18$&	$10.89$&	$16.58$&	$\mathbf{10.44}$\\
		$f_7$&	$5$&	$3$&	$13.01$&	$10.83$&	$9.58$&	$9.27$&	$9.27$&	$13.92$&	$\mathbf{9.21}$\\
		$f_7$&	$5$&	$5$&	$4.58$&	$2.72$&	$2.19$&	$3.64$&	$2.99$&	$\mathbf{2.06}$&	$2.85$\\
		$f_7$&	$8$&	$2$&	$46.21$&	$38.35$&	$49.05$&	$49.92$&	$36.29$&	$\mathbf{32.98}$&	$52.02$\\
		$f_7$&	$8$&	$3$&	$28.11$&	$22.07$&	$\mathbf{21.57}$&	$24.58$&	$26.87$&	$29.28$&	$23.72$\\
		$f_7$&	$8$&	$5$&	$22.44$&	$21.30$&	$22.97$&	$21.16$&	$20.79$&	$\mathbf{20.40}$&	$21.06$\\
		\hline
	\end{tabular}
\end{table}

In conclusion, for the objective of minimizing RMSE, we recommend the adoption of the gradient-based criterion. Conversely, for the objective of minimizing MAE, we recommend the adoption of the variance-based criterion. If both the one-point sequential design and the batch sequential design are feasible, we recommend the cluster-based top-$b$ sequential design from Algorithm \ref{alg:2}, as it not only reduces computational resource consumption but also achieves or even surpasses the performance of one-point sequential design.

\section{Concluding remarks}
\label{sec:conc}
In this paper, focusing on the sequential design of Kriging models for computer experiments, we introduce two novel criteria: gradient-based criterion and variance-based criterion, to guide the selection of additional design points in iterative sampling processes. We further provide accompanying theoretical justifications and algorithmic implementations to underpin the rationality and operability of these criteria. Specifically, the gradient-based criterion is more suitable for minimizing the average error, while the variance-based criterion is more appropriate for minimizing the maximum error. For batch sequential design, we put forward a framework termed the cluster-based top-$b$ sequential design approach, which is applicable to any one-point sequential design criterion. Finally, through comparisons with other existing criteria and methods, we demonstrate that the proposed criteria exhibit excellent performance in one-point sequential design. Moreover, applying them within the proposed framework also yields favorable results through the comparisons.

\bmhead{Acknowledgements} This research was supported by the National Natural Science Foundation of China (Grant Nos. 12131001 and 12371260) and the Fundamental Research Funds for Central Universities.

\section*{Declarations}
The authors declare that there is no conflict of interest. 


\begin{appendices}

\section*{Appendix: proofs}\label{secA1}

\begin{proof}[Proof of Theorem \ref{the:1}]
	We first note that
	$$|f(\boldsymbol{x})-\hat{y}_{\boldsymbol{X}}(\boldsymbol{x})| = |f(\boldsymbol{x})-f(\boldsymbol{x}^*)+f(\boldsymbol{x}^*)-\hat{y}_{\boldsymbol{X}}(\boldsymbol{x})|\leq|f(\boldsymbol{x})-f(\boldsymbol{x}^*)|+|f(\boldsymbol{x}^*)-\hat{y}_{\boldsymbol{X}}(\boldsymbol{x})|.$$
	Since $f(\boldsymbol{x})$ is smooth, by the Lagrange Mean Value Theorem, there exists $\boldsymbol{t}\in\{\alpha\boldsymbol{x}+(1-\alpha)\boldsymbol{x}^*\,|\,0\leq\alpha\leq1\}$ such that $|f(\boldsymbol{x})-f(\boldsymbol{x}^*)| = \nabla f(\boldsymbol{t}) \cdot (\boldsymbol{x}-\boldsymbol{x}^*)$, where $\cdot$ denotes the inner product. In addition, we have $\nabla f(\boldsymbol{t}) \cdot (\boldsymbol{x}-\boldsymbol{x}^*) \leq \|\nabla f(\boldsymbol{t})\| \times d(\boldsymbol{x},\boldsymbol{x}^*)$. This completes the proof.
\end{proof}

Before proving Theorem \ref{the:2}, we first introduce a lemma. This lemma is derived from the Theorem 11.4 in \cite{Wendland:2004}.
\renewcommand{\thelemma}{A1}
\begin{lemma}
	\label{lem:A1}
	Let $\mathcal{X}\subseteq\mathbb{R}^m$ be open, $\mathcal{N}_k(\mathcal{X})$ be the native Hilbert function space corresponding to the symmetric positive definite kernel $k(\boldsymbol{x},\boldsymbol{y})\subseteq C^{2m}(\mathcal{X}, \mathcal{X}):\mathcal{X}\times\mathcal{X}\rightarrow\mathbb{R}$ and $f$ be the smooth true function. Then for every $\boldsymbol{x} \in \mathcal{X}$ and every $\alpha \in \mathbb{N}_0^m$ with $\mathbf{1}^T\alpha\leq m$, we have
	$$
	\left| D^\alpha f(\boldsymbol{x}) - D^\alpha \hat{y}_{\boldsymbol{X}}(\boldsymbol{x}) \right| \leq P_{k,\boldsymbol{X}}^{(\alpha)}(\boldsymbol{x}) | f |_{\mathcal{N}_k(\mathcal{X})},$$
	where $$\left[ P_{k,\boldsymbol{X}}^{(\alpha)}(\boldsymbol{x}) \right]^2 = D_1^\alpha D_2^\alpha k(\boldsymbol{x}, \boldsymbol{x}) - 2\sum_{j=1}^n D^\alpha u_j^*(\boldsymbol{x}) D_1^\alpha k(\boldsymbol{x}, \boldsymbol{x}_j) + \sum_{i,j=1}^n D^\alpha u_i^*(\boldsymbol{x}) D^\alpha u_j^*(\boldsymbol{x}) k(\boldsymbol{x}_i, \boldsymbol{x}_j)$$
	and $u_i^*(\boldsymbol{x})$ is the $i$-th component of $\boldsymbol{K}^{-1}(\boldsymbol{X})\boldsymbol{r}_{\boldsymbol{X}}(\boldsymbol{x}).$ 
\end{lemma}

\begin{proof}[Proof of Theorem \ref{the:2}]
	1. Let $\alpha$ be an $m$-dimensional vector with all zeros, then we have
	\begin{align*}
		\left[ P_{k,\boldsymbol{X}}^{(\alpha)}(\boldsymbol{x}) \right]^2 &= k(\boldsymbol{x}, \boldsymbol{x}) - 2\sum_{j=1}^n u_j^*(\boldsymbol{x}) k(\boldsymbol{x}, \boldsymbol{x}_j) + \sum_{i,j=1}^n u_i^*(\boldsymbol{x}) u_j^*(\boldsymbol{x}) k(\boldsymbol{x}_i, \boldsymbol{x}_j)\\&=k(\boldsymbol{x}, \boldsymbol{x}) - 2\boldsymbol{r}^T_{\boldsymbol{X}}(\boldsymbol{x})\boldsymbol{K}^{-1}(\boldsymbol{X})\boldsymbol{r}_{\boldsymbol{X}}(\boldsymbol{x}) + \boldsymbol{r}^T_{\boldsymbol{X}}(\boldsymbol{x})\boldsymbol{K}^{-1}(\boldsymbol{X})\boldsymbol{K}(\boldsymbol{X})\boldsymbol{K}^{-1}(\boldsymbol{X})\boldsymbol{r}_{\boldsymbol{X}}(\boldsymbol{x})\\&=k(\boldsymbol{x},\boldsymbol{x})-\boldsymbol{r}_{\boldsymbol{X}}^T(\boldsymbol{x})\boldsymbol{K}^{-1}(\boldsymbol{X})\boldsymbol{r}_{\boldsymbol{X}}(\boldsymbol{x}).
	\end{align*}
	According to Lemma \ref{lem:A1}, we have $\left| f(\boldsymbol{x}) - \hat{y}_{\boldsymbol{X}}(\boldsymbol{x}) \right| \leq \sqrt{k(\boldsymbol{x},\boldsymbol{x})-\boldsymbol{r}_{\boldsymbol{X}}^T(\boldsymbol{x})\boldsymbol{K}^{-1}(\boldsymbol{X})\boldsymbol{r}_{\boldsymbol{X}}(\boldsymbol{x})} \times | f |_{\mathcal{N}_k(\mathcal{X})}$. This completes the proof.
	
	2. Let $h(\boldsymbol{x}) = f(\boldsymbol{x}) - \hat{y}_{\boldsymbol{X}}(\boldsymbol{x})$. Then we have
	$h(\boldsymbol{x}) = h(\boldsymbol{x}^*) + \nabla h(\boldsymbol{x}) \cdot (\boldsymbol{x}-\boldsymbol{x}^*) + O(\boldsymbol{x}-\boldsymbol{x}^*)^2$. It is obvious that $h(\boldsymbol{x}^*)=0$. Similar to the proof of Theorem 1, we have $|h(\boldsymbol{x})|\leq \| \nabla h(\boldsymbol{t})\|\times d(\boldsymbol{x},\boldsymbol{x}^*),$ where $\boldsymbol{t}\in\{\alpha\boldsymbol{x}+(1-\alpha)\boldsymbol{x}^*\,|\,0\leq\alpha\leq1\}$ and $ \nabla h(\boldsymbol{t}) = (\partial h(\boldsymbol{t})/\partial t_1, \ldots, \partial h(\boldsymbol{t})/\partial t_m)$. Let $\alpha_i$ be an $m$-dimensional vector whose $i$-th element is $1$ and the remaining elements are $0$. According to Lemma \ref{lem:A1}, we have $|\partial h(\boldsymbol{t})/\partial t_i|\leq P_{k,\boldsymbol{X}}^{(\alpha_i)}(\boldsymbol{t}) | f |_{\mathcal{N}_k(\mathcal{X})}$. That is, 
	$$\| \nabla h(\boldsymbol{t})\|=\sqrt{\sum_{i=1}^{m}\left(\frac{\partial h(\boldsymbol{t})}{\partial t_i}\right)^2}\leq\sqrt{\sum_{i=1}^{m}\left[P_{k,\boldsymbol{X}}^{(\alpha_i)}(\boldsymbol{t}) \right]^2}\times | f |_{\mathcal{N}_k(\mathcal{X})}.$$
	
	According to the rules of differentiation for vectors and matrices, we know that $\left[P_{k,\boldsymbol{X}}^{(\alpha_i)}(\boldsymbol{t}) \right]^2$ is the $i$-th diagonal element of  
	\begin{align*}
		&\nabla^2k(\boldsymbol{t},\boldsymbol{t})-2\nabla r_{\boldsymbol{X}}^T(\boldsymbol{t})\boldsymbol{K}^{-1}(\boldsymbol{X})\nabla r_{\boldsymbol{X}}(\boldsymbol{t})+\nabla r_{\boldsymbol{X}}^T(\boldsymbol{t})\boldsymbol{K}^{-1}(\boldsymbol{X})\boldsymbol{K}(\boldsymbol{X})\boldsymbol{K}^{-1}(\boldsymbol{X})\nabla r_{\boldsymbol{X}}(\boldsymbol{t})\\=&\nabla^2k(\boldsymbol{t},\boldsymbol{t})-\nabla r_{\boldsymbol{X}}^T(\boldsymbol{t})\boldsymbol{K}^{-1}(\boldsymbol{X})\nabla r_{\boldsymbol{X}}(\boldsymbol{t}).
	\end{align*}
	Combining the above results, we can obtain that $\| \nabla h(\boldsymbol{t})\|\leq\sqrt{\sum_{i=1}^mg_i(\boldsymbol{t})}\times | f |_{\mathcal{N}_k(\mathcal{X})}$. This completes the proof. 
\end{proof}




\end{appendices}


\bibliography{SP-bib}

\end{document}